\def\UseBibLatex{1}
\makeatletter
\def\input@path{{styles/}}
\makeatother

\providecommand{\pgfsyspdfmark}[3]{}

\documentclass[12pt]{article}
\usepackage{lmodern}
\usepackage[T1]{fontenc}

\providecommand{\BibLatexMode}[1]{}
\providecommand{\BibTexMode}[1]{}

\renewcommand{\BibLatexMode}[1]{}
\renewcommand{\BibTexMode}[1]{#1}

\usepackage[bibencoding=utf8,style=alphabetic,backend=biber]{biblatex}%
\usepackage{sariel_biblatex}%

\usepackage{amsmath}%
\usepackage{amssymb}%
\usepackage[table]{xcolor}%

\usepackage{microtype}
\usepackage{xfrac}

\usepackage[amsmath,thmmarks]{ntheorem}%
\usepackage{iftex}

\usepackage{titlesec}%
\usepackage{xcolor}%
\usepackage{mleftright}%
\usepackage{xspace}%
\usepackage{graphicx}
\usepackage{hyperref}%
\usepackage[inline]{enumitem}
\usepackage{hyperref}%
\usepackage[ocgcolorlinks]{ocgx2}
\usepackage{comment}

\usepackage{mathcalb}
\usepackage{euscript}%

\definecolor{darkyred}{rgb}{0.25, 0, 0}

\hypersetup{%
      unicode,
      breaklinks,%
      colorlinks=true,%
      urlcolor=darkyred,%
      linkcolor=[rgb]{0.5,0.0,0.0},%
      citecolor=[rgb]{0,0.2,0.445},%
      filecolor=[rgb]{0,0,0.4},%
      anchorcolor=[rgb]{0.0,0.1,0.2},%
      hypertexnames=false
   }%

\titlelabel{\thetitle. }%

\theoremseparator{.}%

\theoremstyle{plain}%
\newtheorem{theorem}{Theorem}[section]

\theoremseparator{}%

\theoremseparator{.}%

\newtheorem{lemma}[theorem]{Lemma}

\newtheorem{corollary}[theorem]{Corollary}

\theoremstyle{plain}%
\theoremheaderfont{\sf}
\theorembodyfont{\upshape}%
\newtheorem*{remark:unnumbered}[theorem]{Remark}%
\newtheorem{remark}[theorem]{Remark}%
\newtheorem{tedium}[theorem]{Tedium}%
\newtheorem{observation}[theorem]{Observation}
\newtheorem{definition}[theorem]{Definition}
\newtheorem{defn}[theorem]{Definition}

\newtheorem{problem}[theorem]{Problem}

   \theoremheaderfont{\em}%
   \theorembodyfont{\upshape}%
   \theoremstyle{nonumberplain}%
   \theoremseparator{}%
   \theoremsymbol{\myqedsymbol}%
   \newtheorem{proof}{Proof:}%

\providecommand{\emphind}[1]{}%
\renewcommand{\emphind}[1]{\emph{#1}\index{#1}}

\definecolor{blue25emph}{rgb}{0, 0, 11}

\providecommand{\emphic}[2]{}
\renewcommand{\emphic}[2]{\textcolor{blue25emph}{%
      \textbf{\emph{#1}}}\index{#2}}

\providecommand{\emphi}[1]{}%
\renewcommand{\emphi}[1]{\emphic{#1}{#1}}

\definecolor{almostblack}{rgb}{0, 0, 0.3}

\providecommand{\emphw}[1]{}%
\renewcommand{\emphw}[1]{{\textcolor{almostblack}{\emph{#1}}}}%

\providecommand{\emphOnly}[1]{}%
\renewcommand{\emphOnly}[1]{\emph{\textcolor{blue25emph}{\textbf{#1}}}}

\newcommand{\myqedsymbol}{\rule{2mm}{2mm}}
\newcommand{\SarielThanks}[1]{%
    \thanks{%
       Department of Computer Science; %
       University of Illinois; %
       201 N. Goodwin Avenue; %
       Urbana, IL, 61801, USA; %
       \href{mailto:spam@illinois.edu}{sariel@illinois.edu}; %
       \url{http://sarielhp.org/}.%
    #1%
    }%
}

\newcommand{\HLinkB}[2]{\hyperref[#2]{#1}}
\newcommand{\HLinkD}[2]{\hyperref[#2]{\ref*{#2}}}
\newcommand{\HLink}[2]{\hyperref[#2]{#1~\ref*{#2}}}
\newcommand{\HLinkSuffix}[3]{\hyperref[#2]{#1\ref*{#2}{#3}}}

\newcommand{\figlab}[1]{\label{fig:#1}}
\newcommand{\figref}[1]{\HLink{Figure}{fig:#1}}
\newcommand{\figrefX}[1]{\HLinkD{}{fig:#1}}

\newcommand{\thmlab}[1]{{\label{theo:#1}}}
\newcommand{\thmref}[1]{\HLink{Theorem}{theo:#1}}

\newcommand{\corlab}[1]{\label{cor:#1}}

\newcommand{\seclab}[1]{\label{sec:#1}}
\newcommand{\secref}[1]{\HLink{Section}{sec:#1}}

\providecommand{\deflab}[1]{\label{def:#1}}
\newcommand{\defref}[1]{\HLink{Definition}{def:#1}}
\newcommand{\defrefY}[2]{\hyperref[def:#1]{#2}}

\newcommand{\apndlab}[1]{\label{apnd:#1}}

\newcommand{\tbllab}[1]{\label{table:#1}}
\newcommand{\tblref}[1]{\HLink{Table}{table:#1}}
\newcommand{\tblrefX}[1]{\HLinkD{}{table:#1}}

\newcommand{\lemlab}[1]{\label{lemma:#1}}
\newcommand{\lemref}[1]{\HLink{Lemma}{lemma:#1}}%

\newcommand{\tedlab}[1]{\label{tedium:#1}}
\newcommand{\tedref}[1]{\HLink{Tedium}{tedium:#1}}%

\providecommand{\eqlab}[1]{}%
\renewcommand{\eqlab}[1]{\label{equation:#1}}

\providecommand{\remove}[1]{}%
\newcommand{\Set}[2]{\left\{ #1 \;\middle\vert\; #2 \right\}}

\newcommand{\pth}[1]{\mleft(#1\mright)}%

\newcommand{\ceil}[1]{\mleft\lceil {#1} \mright\rceil}

\newcommand{\brc}[1]{\left\{ {#1} \right\}}
\newcommand{\set}[1]{\brc{#1}}%

\newcommand{\cardin}[1]{\left\lvert {#1} \right\rvert}%

\renewcommand{\th}{th\xspace}

\renewcommand{\Re}{\mathbb{R}}%

\newlist{compactenumA}{enumerate}{5}%
\setlist[compactenumA]{topsep=0pt,itemsep=-1ex,partopsep=1ex,parsep=1ex,%
   label=(\Alph*)}%

\newlist{compactenuma}{enumerate}{5}%
\setlist[compactenuma]{topsep=0pt,itemsep=-1ex,partopsep=1ex,parsep=1ex,%
   label=(\alph*)}%

\newlist{compactenumI}{enumerate}{5}%
   \setlist[compactenumI]{topsep=0pt,itemsep=-1ex,partopsep=1ex,parsep=1ex,%
   label=(\Roman*)}%
\newlist{compactenumIn}{enumerate*}{5}%
\setlist[compactenumIn]{%
   label=(\Roman*)}%

\newlist{compactenumin}{enumerate*}{5}%
\setlist[compactenumin]{label=(\roman*)}%

\newlist{compactenumi}{enumerate}{5}%
\setlist[compactenumi]{topsep=0pt,itemsep=-1ex,partopsep=1ex,parsep=1ex,%
   label=(\roman*)}%

\newlist{compactitem}{itemize}{5}%
\setlist[compactitem]{topsep=0pt,itemsep=-1ex,partopsep=1ex,parsep=1ex,%
   label=\ensuremath{\bullet}}%

   \numberwithin{figure}{section}%
   \numberwithin{table}{section}%
   \numberwithin{equation}{section}%

\usepackage{todonotes}

\renewcommand{\P}{P}%
\newcommand{\Q}{Q}%
\newcommand{\eps}{\varepsilon}%
\newcommand{\epsA}{\xi}%

\newcommand{\Term}[1]{\textsf{#1}}

\definecolor{vdarkyred}{rgb}{0.1, 0, 0}
\newcommand{\ProblemC}[1]{{\textsf{\textcolor{vdarkyred}{#1}}}}

\newcommand{\SetCover}{\ProblemC{Set{}Cover}\xspace}
\newcommand{\VertexCover}{\ProblemC{Vertex{}Cover}\xspace}

\newcommand{\CliqueCover}{\ProblemC{Clique{}Cover}\xspace}
\newcommand{\EdgeCliqueCover}{\ProblemC{EClique{}Cover}\xspace}

\providecommand{\PairwiseClustering}{\ProblemC{Pairwise{}\-Clustering}\xspace}

\newcommand{\Eps}{\ensuremath{\tfrac{1}{\eps}}\xspace}
\newcommand{\EpsA}{\ensuremath{\tfrac{1}{\epsA}}\xspace}
\newcommand{\WSPD}{\Term{WSPD}\xspace}
\newcommand{\SSPD}{\Term{SSPD}\xspace}
\newcommand{\eSSPD}{\ensuremath{\Eps}-\Term{SSPD}\xspace}

\newcommand{\cWSPD}{\Term{CoverWSPD}\xspace}
\newcommand{\pWSPD}{\Term{PartitionWSPD}\xspace}
\newcommand{\MST}{\Term{MST}\xspace}
\newcommand{\PTAS}{\Term{PTAS}\xspace}

\newcommand{\eWSPD}{\ensuremath{\tfrac{1}{\eps}}-\WSPD{}\xspace}
\newcommand{\ecWSPD}{\ensuremath{\tfrac{1}{\eps}}-\cWSPD{}\xspace}
\newcommand{\etal}{\textit{et~al.}\xspace}

\definecolor{VDGreen}{rgb}{0.0, 0.10, 0}

\newcommand{\TextX}[1]{%
   \begin{minipage}{0.18\linewidth}
       #1
   \end{minipage}
}

\newcommand{\InsertFigure}[7]{%
\begin{figure}[H]
    \phantom{}%
    \hfill%
    \includegraphics[page=2,width=0.18\linewidth]{figs/#2}
    \hfill
    \includegraphics[page=4,width=0.18\linewidth]{figs/#2}
    \hfill%
    \includegraphics[page=6,width=0.18\linewidth]{figs/#2}
    \hfill%
    \includegraphics[page=8,width=0.18\linewidth]{figs/#2}
    \hfill%
    \includegraphics[page=10,width=0.18\linewidth]{figs/#2}
    \hfill%
    \phantom{}

    \phantom{}
    \hfill%
    \TextX{Greedy: #6}%
    \hfill%
    \TextX{\WSPD: #3}%
    \hfill%
    \TextX{\AprxT: #4}%
    \hfill%
    \TextX{\AprxTC: #5}%
    \hfill%
    \TextX{\IP: #7}%
    \hfill%
    \phantom{}

    \caption{Solution for $n={#1}$, the input is $1,2, \ldots , n$, and $\eps=1.0$.}
    \figlab{#2}
\end{figure}
}

\DefineNamedColor{named}{OliveGreen} {cmyk}{0.64,0,0.95,0.20} \providecommand{\ComplexityClass}[1]{%
   {{\textnormal{\textsc{%
               \textcolor[named]{VDGreen}{%
                  #1}}}}}%
}%
\renewcommand{\ComplexityClass}[1]{{{\textcolor[named]{VDGreen}{%
            \textnormal{\textsc{#1}}}}}}%
\providecommand{\NPHardness}{{\ComplexityClass{NP-Hardness}}\xspace}%
\providecommand{\NPHard}{{\ComplexityClass{NP-Hard}}%
   \index{NP!hard}\xspace}%
  \providecommand{\NP}{\ComplexityClass{NP}%
   \index{NP}\xspace}  \providecommand{\POLYT}{\ComplexityClass{P}\xspace}

\newcommand{\dC}{\mathcalb{d}}%
\newcommand{\DistChar}{\dC}%
\newcommand{\dmY}[2]{\DistChar\pth{#1,#2}}%

\newcommand{\dminY}[2]{\dC_{\min}\pth{#1,#2}}
\newcommand{\dmaxY}[2]{\dC_{\max}\pth{#1,#2}}

\newcommand{\G}{G}
\newcommand{\VV}{\mathsf{V}}
\newcommand{\Edges}{\mathsf{E}}

\newcommand{\VX}[1]{\VV\pth{#1}}
\newcommand{\EGX}[1]{\Edges\pth{#1}}

\newcommand{\WR}{\mathcal{W}}

\newcommand{\diamC}{\nabla}%
\newcommand{\diamX}[1]{\diamC\pth{#1}}

\newcommand{\ts}{\hspace{0.6pt}}
\newcommand{\FMS}{\EuScript{X}}%

\newcommand{\Spread}{\Psi}%
\newcommand{\SpreadX}[1]{\Spread\pth{#1}}%
\newcommand{\cpX}[1]{\mathrm{cp}\pth{#1}}
\newcommand{\pair}{\mathcalb{p}}%
\newcommand{\opair}{\mathcalb{o}}%
\newcommand{\F}{\mathcal{F}}%
\providecommand{\G}{\mathcal{G}}%
\renewcommand{\G}{\mathcal{G}}%

\newcommand{\pairsX}[1]{E_{#1}}%

\newcommand{\SQ}{\mathcal{S}}%

\newcommand{\lenX}[1]{\ell\pth{#1}}%

\newcommand{\R}{\mathcal{R}}%
\newcommand{\V}{\mathcal{V}}%
\newcommand{\rect}{\mathcalb{r}}%

\usepackage{stmaryrd}%
\providecommand{\IntRange}[1]{\mleft\llbracket #1 \mright\rrbracket}
\newcommand{\IRX}[1]{\IntRange{#1}}%

\newcommand{\minWSPD}{\Term{minWSPD}\xspace}

\newcommand{\shadowX}[1]{\mathsf{S}_{#1}}

\newcommand{\sq}{\mathcalb{s}}

\newcommand{\OPT}{\mathcal{O}}%

\newcommand{\Opt}{\mathrm{opt}}
\newcommand{\UU}{\mathcal{U}}%

\newcommand{\ball}{\mathcalb{b}}%
\newcommand{\ballY}[2]{\ball\pth{#1,#2}}%
\newcommand{\Rep}{\rho}
\newcommand{\repX}[1]{\Rep\pth{#1}}
\newcommand{\mlevelX}[1]{L({#1})}

\newcommand{\Odim}{{O(\dim)}}

\newcommand{\ZZ}{\mathbb{Z}}%

\newcommand{\parent}{\overline{\mathrm{p}}}
\newcommand{\tbase}{\tau}

\newcommand{\nettree}{net-tree\xspace}

\newcommand{\Nettrees}{Net-trees\xspace}

\newcommand{\rootX}[1]{\mathcal{r}\pth{#1}}

\newcommand{\hcite}[2][]{\emph{\textbf{\cite[#1]{#2}.}}}

\newcommand{\KCenPrcY}[2]{\dC_{\max}\pth{#1,#2}}

\newcommand{\KCenPrcOptY}[2]{\Opt_{#2}\pth{#1}}
\newcommand{\optCDiamY}[2]{\Opt_{\diamC,#1}\pth{#2}}
\newcommand{\NN}{\mathbb{N}}
\newcommand{\CS}{C}

\ifpdftex
\fi

\usepackage{bold-extra}

\newcommand{\SQSet}{\mathcal{C}}

\newcommand{\Julia}{\texttt{Julia}\xspace}%
\newcommand{\Gurobi}{\texttt{Gurobi}\xspace}%
\newcommand{\IP}{\Term{IP}\xspace}%
\newcommand{\AprxT}{\Term{Aprx3}\xspace}%
\newcommand{\AprxTC}{\Term{Aprx3C}\xspace}%

\usepackage{pict2e}%

\newcommand{\ABC}{\Term{ABC}\xspace}
\newcommand{\eABC}{\ensuremath{\tfrac{1}{\eps}}-\ABC{}\xspace}

\newcommand{\ringC}{%
   \begin{picture}(10,10)
       \color{gray}
       \linethickness{2pt}
       \put(5,3.4){\circle{6}}
   \end{picture}%
}

\newcommand{\ringZ}[3]{\ringC\pth{#1,#2,#3}}

\newcommand{\PB}{X}%
\newcommand{\PC}{Y}%

\newcommand{\pb}{x}
\newcommand{\pc}{y}

\newcommand{\obslab}[1]{\label{observation:#1}}
\newcommand{\obsref}[1]{\HLink{Observation}{observation:#1}}
\newcommand{\Cube}{\mathcal{C}}%
\newcommand{\num}{\zeta}%
\newcommand{\Grid}{\mathcal{G}}%
\newcommand{\Cell}{\square}%
\newcommand{\UC}{[0,1)^d}
\newcommand{\shift}{\nu}
\newcommand{\qte}{\mathcal{T}_\eps}
\newcommand{\qteA}{\mathcal{T}_\epsA}

\newcommand{\Pfar}{P_{\mathrm{far}}}

\newcommand{\cen}{\overline{\mathsf{c}}}
\newcommand{\cenX}[1]{\overline{\mathsf{c}}\pth{#1}}

\newcommand{\weightX}[1]{\mathrm{w}\pth{#1}}%

\newcommand{\normX}[1]{\left\| {#1} \right\|}
\newcommand{\dY}[2]{\normX{#1 - #2}}

\newcommand{\LgEpsA}{\lambda}
\newcommand{\qt}{\mathcal{T}}%

\BibTexMode{%
   \newcommand{\tcite}[1]{\cite{#1}}
}

\newcommand{\Saarbrucken}{Saarbr\"{u}cken\xspace}

\usepackage{placeins}%

\newcommand{\prflab}[1]{\label{proof:#1}}

\bibliography{wspd_min}

\usepackage{float}

\newcommand{\sparagraph}[1]{\paragraph*{#1}}

\usepackage{bbding}

\definecolor{darkgreen}{rgb}{0,0.10,0}

\newcommand{\breakhere}{\hspace{0pt}\allowbreak}

\newcommand{\Correct}{\textcolor{darkgreen}{\ensuremath{\checkmark}}}
\newcommand{\Wrong}{\textcolor{red}{{\XSolidBrush}}}

\usepackage[all]{hypcap}
\usepackage{placeins}
\usepackage[export]{adjustbox}
\newcommand{\fstabC}{\mathsf{S}}
\newcommand{\fstabX}[1]{\fstabC\pth{#1}}
\newcommand{\fstabY}[2]{\fstabC\pth{#1,#2}}

\newcommand{\SpellIgnore}[1]{#1}

\title{On Small Pair Decompositions for Point Sets}

   \author{%
      Kevin Buchin%
      \thanks{Institute of Computer Science, TU Dortmund University, Germany}%
      \and%
      Jacobus Conradi%
      \thanks{Institute of Computer Science, Universität Bonn, Germany and University of Copenhagen, Denmark}%
      \and%
      Sariel Har-Peled%
      \SarielThanks{}%
      \and%
      Antonia Kalb%
      \thanks{ TU Dortmund University, Germany.  \href{mailto:spam@spamalot.com}{antonia.kalb@tu-dortmund.de}.  \url{https://orcid.org/0009-0009-0895-8153}.  }%
      \and%
      Abhiruk Lahiri%
      \thanks{Heinrich Heine University, \SpellIgnore{D\"{u}sseldorf} and CISPA Helmholtz Centre for Information Security, \Saarbrucken, Germany}%
      \and%
      Lukas Plätz%
      \thanks{Institute of Computer Science, Ruhr University Bochum, Germany}%
      \and%
      Carolin Rehs%
      \thanks{%
         Department of Mathematics and Computing Science, TU Eindhoven, the Netherlands and Institute of Computer Science, TU Dortmund University, Germany \href{mailto:spam@spamalot.com}{carolin.rehs@tu-dortmund.de}.  \url{https://orcid.org/0000-0002-8788-1028} Supported by the PRIME \SpellIgnore{programme} of the German Academic Exchange Service (D{AA}D) with funds from the German Federal Ministry of Research, Technology and Space (BM{F}T{R})%
   }%
   \and%
   Sampson Wong%
   \thanks{Department of Computer Science, University of Copenhagen, Denmark}
   }%

\begin{document}

\date{\today}%

\maketitle

\begin{abstract}
    We study the \minWSPD problem of computing the minimum-size well-separated pairs decomposition of a set of points, and show constant approximation algorithms in low-dimensional Euclidean space and doubling metrics. This problem is computationally hard already $\Re^2$, and is also hard to approximate.

    We also introduce a new pair decomposition, removing the requirement that the diameters of the parts should be small.  Surprisingly, we show that in a general metric space, one can compute such a decomposition of size $O( \tfrac{n}{\eps}\log n)$, which is dramatically smaller than the quadratic bound for \WSPD{s}. In $\Re^d$, the bound improves to $O( d \tfrac{n}{\eps}\log \tfrac{1}{\eps } )$.
\end{abstract}

\section{Introduction}

Given a set $\P$ of $n$ points, the \emph{Well-Separated Pair Decomposition} (\WSPD), introduced by Callahan and Kosaraju \cite{ck-dmpsa-95}, is an approximate compact representation of the metric induced by $P$. It has been used to solve numerous problems, such as the all-nearest-neighbors \cite{ck-dmpsa-95}, approximate \MST~\cite{ck-dmpsa-95}, spanners \cite{ns-gsn-07}, Delaunay triangulations~\cite{bm-dt-11}, distance oracles \cite{gz-wspdu-05}, and others \cite{hm-fcnld-06}. See Smid \cite{s-wspda-18} for a survey.

For a prespecified $\eps \in (0,1)$, the \eWSPD represents the $\binom{n}{2}$ pairwise distances, of a set $\P \subseteq \Re^d$ of $n$ points, as a union of $O(n/\eps^d)$ (edge-disjoint) bicliques, where all the distances in a single biclique are the same up to a multiplicative factor in $1\pm O(\eps)$. More precisely, for a pair $\{\PB,\PC\}$ in a \eWSPD, we require that
\begin{equation*}
    \max\bigl( \diamX{\PB}, \diamX{\PC} \bigr) \leq \eps \dmY{\PB}{\PC},
\end{equation*}
where $\diamX{\PB}$ is the diameter of $\PB$, and $\dmY{\PB}{\PC} = \min_{\pb \in \PB, \pc \in \PC} \dmY{\pb}{\pc}$, where $\dmY{\pb}{\pc}$ denotes the distance between $\pb$ and $\pc$.  In words, for a biclique $ X \otimes Y$ in this cover, the diameters of the two sets are at most $\eps$-fraction of their distance from each other, see \defref{WSPD}.  This biclique cover can be computed, in $O(n \log n + n/\eps^d)$ time, where each biclique is represented as a pair of nodes in the tree, see \cite{ck-dmpsa-95}.  The \emphw{size} of the \WSPD is the number of bicliques.

The natural question is how to construct a small (or even the smallest) representation of the pairwise distances in $P$? Specifically, what is the smallest number of bicliques required to cover the $\binom{n}{2}$ pairwise distances in $\binom{\P}{2}$, so that distances in a biclique are roughly the same?  Our approach to answering the above question is first to consider how to minimize the size of the \WSPD, and secondly to introduce a relaxation of the \WSPD and construct even smaller covers, where the diameter requirement on the biclique sides is removed.

\sparagraph{Minimum-size \WSPD{s}.}

The bound on the size of a \eWSPD in $\Re^d$ is $O(n / \eps^d)$ \cite{ck-dmpsa-95}. However, this bound can be far from optimal. In $\Re^2$, the point set $\P = \Set{ p_i = \bigl( i,(4/\eps)^i \bigr) }{ i=1,\ldots, n}$ has a \eWSPD formed by the pairs $\PB_i = \{p_1, \ldots, p_{i}\}$ and $\PC_i = \{p_{i+1}\}$, for $i=2,\ldots, n$. Its size is $(n-1)$, which is significantly smaller than the worst-case bound $\Theta(n/\eps^2)$. Computing small \WSPD{s} may have downstream effects on its applications, such as the sizes of \WSPD-based spanners or distance oracles. This effect motivates us to study the \minWSPD problem.

\begin{problem}[\minWSPD]
    \label{problem:1}
    Given a set $P$ of $n$ points in a metric space $(\FMS,\dC)$, and a parameter $\eps$, compute the minimum-size \eWSPD ($\WR$) of $P$. That is, minimize the number of pairs in $\WR$, while guaranteeing that the resulting \WSPD provides the desired separation.
\end{problem}

\sparagraph{Semi-separated pairs-decomposition (\SSPD).} %
The requirements for a \WSPD are, in some regards, quite stringent. Thus, one may be interested in obtaining smaller pair decompositions by relaxing the requirements of the \WSPD. A neat example for this is semi-separated pair decomposition (\SSPD) \cite{v-dcamc-98, acfs-psspd-09, adfg-rftgs-09, adfgs-gswps-11,ah-ncsa-12}, which requires only one of the sides in each pair to be small.  Formally, for a pair $\{\PB,\PC\}$ in \eSSPD, we require that $\min\bigl( \diamX{\PB}, \diamX{\PC} \bigr) \leq \eps \dmY{\PB}{\PC}$ (note, the $\min$ here instead of $\max$ used in \WSPD). Thus, a semi-separated pair might have a ``tiny'' set (say) $\PB$ on one side, while the other side diameter might be unbounded in terms of the distance to the other set, see \figref{stability}.  For $\Re^d$, there are constructions of \eSSPD with $O(n/\eps^d)$ pairs, with the striking property that their weight is $O( \eps^{-d} n \log n)$. The weight of a pair decomposition is the total size of the sets used in it. In contrast, the weight of the \WSPD, in the worst case, is quadratic.

\sparagraph{Approximate Biclique Covers (\eABC).} %
For Approximate Biclique Covers (\ABC{s}), we are interested in a different relaxation -- what if we do not care about the diameters of the two sides of the pairs at all, as long as the distances are roughly the same? That is, for a pair $\{X,Y \}$ in the pair decomposition, the requirement is that $\dmY{x}{y} \leq (1+\eps)\dmY{X}{Y}$, for all $x \in X, y \in Y$. Such a pair $\{X,Y\}$ is \emphw{$\eps$-stable}. These three different concepts of separation are inherently different, see \figref{stability}.  This decomposition still provides sufficient information to efficiently approximate quantities such as the diameter, \MST, or the closest pair.  Specifically, small \ABC{}s should be useful in speeding up approximate matching algorithms.  More generally, one can try and use the new \ABC decompositions as a replacement for either \WSPD or \SSPD.

We are unaware of any previous work on Approximate Biclique Covers (\ABC{s}) for metric spaces (covering graphs by bi/cliques is, of course, a well-known problem). Of particular interest is whether one can obtain a small \ABC in metrics that do not admit small \WSPD{s}. In general metrics, the \WSPD can have quadratic size and weight.

\sparagraph{\WSPD{}s in other metrics.}
Although \WSPD{s} may have $\Omega(n^2)$ size in general metrics, they are known to have sub-quadratic size in special cases:
\begin{compactenumIn}%
    \item Callahan and Kosaraju~\cite{ck-dmpsa-95} construct a \eWSPD of size $O(n / \eps^d)$ in $\Re^d$.

    \item Gao and Zhang \cite{gz-wspdu-05} showed how to construct $\tfrac{1}{\eps}$-\WSPD of size $O(\eps^{-4} n \log n)$ for the unit-distance graph for a set $\P$ of $n$ points in the plane (the bound is $O_\eps(n^{2-2/d})$ for $d>2$). The 2d bound was slightly improved to $O( \eps^{-2} n \log n)$ by Har-Peled \etal \cite{hrr-wspdr-25}.

    \item Har-Peled and Mendel \cite{hm-fcnld-06} showed how to compute   \eWSPD{}s, of size $n/\eps^{O(\dim)}$  for spaces with  doubling dimension $\dim$.

    \item Gudmundsson and Wong \cite{gw-wspdl-24} showed a $\tfrac{1}{\eps}$-\WSPD of size $O( \lambda^2 \eps^{-4} n \log n)$ for a $\lambda$-low-density graph.

    \item Deryckere \etal \cite{dgrsw-wsstc-25} showed how to compute
    a linear size \eWSPD for $c$-packed graphs.
\end{compactenumIn}

\sparagraph{\WSPD as a min-covering problem.} %
A \WSPD represents $P \otimes P$ as a union of bicliques. While these bicliques are typically required to be disjoint, applications of the \WSPD do not make use of the disjointness. We therefore also investigate whether dropping this requirement results in significantly smaller \WSPD{}s, referred to as \cWSPD. In contrast, the disjoint version is \pWSPD (see \defref{WSPD}).

The problem of computing such a \cWSPD can be seen as a \SetCover instance, where the sets are all maximal well-separated pairs in $\binom{\P}{2}$. When the number of such pairs is polynomial, this readily leads to $\Theta(\log n)$ approximation of the minimum-size \cWSPD, see \secref{set_cover}. In \secref{1_dim} we present algorithms with better approximation factors for $P \subset \Re$, which in turn also leads to a constant-factor approximation of the minimum-size \pWSPD for such $P$.

\subsection{Our contributions}

This paper presents new results for the \eABC and the minimum \eWSPD.
\begin{compactenumI}[leftmargin=0.7cm]
    \smallskip%
    \item \textsf{\ABC{}s and \SSPD{}s.}  In \secref{abc}, we show, surprisingly, that \emph{any} finite metric space admits an \ABC of near-linear size. Specifically, we construct an \eABC with $O(\frac n \eps \log n)$ size. In contrast to \WSPD{s}, which only have sub-quadratic size in special cases. Furthermore, the resulting decomposition is also an \eSSPD.

    In $\Re^d$, we improve the size bound to $O( d^{3} \tfrac{n}{\eps}\log \tfrac{1}{\eps})$. Happily, the new construction is again an \eSSPD. The construction improves the number of pairs in \SSPD from the old bound of $O(n/\eps^d)$ to the stated bound, which has only polynomial dependency on the dimension $d$ instead of exponential.  While the new \ABC/\SSPD has a compact representation, it has two drawbacks: (i) it is a cover, not a disjoint cover, of the pairs, and (ii) the weight of the decomposition can be quadratic in the worst case.

    \medskip%
    \item \textsf{\WSPD in doubling metrics.}  In \secref{doubling_metrics}, we show that the classical algorithm using net-trees~\cite{hm-fcnld-06} for constructing a \WSPD in metric spaces with bounded doubling dimension is nearly optimal. More specifically, for a set of points $\P$ with constant doubling dimension and $\eps\in(0,1)$ we show that the computed \WSPD has size $O(\Opt)$, where $\Opt$ is the size of the minimum-size \WSPD of $\P$, and the running time of the algorithm is in $O(\Opt +n\log n)$, making its running time optimal output-sensitive.\footnote{It is known that computing the \WSPD requires $\Omega(n \log n)$ time due to a reduction to Nearest Pair.} This is in contrast to the best known running time of the algorithm of $O(\eps^{-d}n+n\log n)$.

    \smallskip%
    \item
    \textsf{Computational hardness.} %
    In \secref{hardness} we study the hardness of \minWSPD.  We show that for general metric spaces, computing the minimum size $3$-\WSPD is \NPHard (\thmref{NP_hard_arbitrary}). Furthermore, a simple padding of this construction shows that the problem is polynomially hard to approximate (\lemref{hard_m_s}).  In the plane, we show that deciding if a \eWSPD has a prespecified size is \NPHard, see \thmref{NP_hard_euclidean_2} (here $\eps$ is quite small, roughly $1/4n$).

    \smallskip%
    \item \textsf{Improvements in one dimension.} %
    In \secref{1_dim}, we study the special case of $\P\subset\Re$. Our first algorithm constructs a \SetCover instance and uses known \PTAS algorithms to obtain a $(1+\delta)$-approximation algorithm with running time $n^{O(1/\delta^{2})}$. Our second algorithm uses the sweep-line algorithm to obtain an easy-to-implement $3$-approximation with running time $O(n\log n)$. Finally, when one requires the separated pairs to be disjoint, we give a constant approximation.

    \smallskip%
    \item \textsf{Experiments.}  %
    In \secref{experiments} we report on our implementation of six \WSPD algorithms, for the one-dimensional case, and compare the sizes of the outputs.  The algorithms we implement are: a Greedy set cover algorithm, the \WSPD algorithm of Callahan and Kosaraju~\cite{ck-dmpsa-95}, the 3-approximation algorithm, the 3-approximation algorithm with a post-\SpellIgnore{proc\-essing} cleanup stage, an \IP solver, and an \IP solver that outputs disjoint covers.

\end{compactenumI}

\sparagraph{Comparison to previous work.}
Computing a \WSPD with instance-specific guarantees has been considered by Har-Peled \etal \cite{hrr-wspdr-25}. Our result differs in that our \eWSPD is a constant factor approximation compared to the minimum \vphantom{$\tfrac{1}{\eps}^{Z}$}\eWSPD, whereas Har-Peled \etal \cite{hrr-wspdr-25} construction size is proportional to the minimum $\tfrac{33}{\eps}^{}$-\WSPD.

\subsubsection*{Paper organization}

In \secref{wspd_def} we cover the preliminaries. In \secref{abc} we introduce the approximate biclique covers and show how to construct such a \SpellIgnore{decom\-positions}.  In \secref{doubling_metrics} we show that the net-tree based \WSPD construction is a constant factor approximation.  In \secref{hardness} we state our NP-hardness lower bounds.  In \secref{1_dim} we present our results for points for the special case of $\Re$.  Specifically, in \secref{1D_WSPD_by_SetCover}, we present a \PTAS by reducing it to \SetCover.  In \secref{1_dim_3_approx} we show an $O(n \log n)$ time algorithm that provides a $3$-approximation. In \secref{partition} we show how to convert a \cWSPD into a \pWSPD.  We present our experiments in \secref{experiments}.  Conclusions are provided in \secref{conclusions}.

\section{Preliminaries}
\seclab{wspd_def}

\subsection{Metrics, distances, balls and spread}

\begin{definition}
    \deflab{metric_space}%
    A \emphi{metric space} $\FMS$ is a pair $(\FMS, \DistChar )$, where $\FMS$ is the ground set and $\DistChar: \FMS \times \FMS \rightarrow [0, \infty)$ is a \emphi{metric} satisfying the following conditions for all $x,y\in \FMS$: (i) $\dmY{x}{y} = 0$ if and only if $x =y$, (ii) $\dmY{x}{y} = \dmY{y}{x}$, and (iii) $\dmY{x}{y} + \dmY{y}{z} \geq \dmY{x}{z}$.

    For a set $\P \subseteq \FMS$, its \emphi{diameter} is
    \begin{equation*}
        \diamX{\P} = \max_{x,y \in \P} \dmY{x}{y}.
    \end{equation*}
    The \emphi{distance} between two sets $\PB, \PC \subseteq \FMS$ is
    \begin{equation*}
        \dmY{\PB}{\PC} = \dminY{\PB}{\PC} =\min_{\pb \in \PB, \pc \in \PC} \dmY{\pb}{\pc}.
    \end{equation*}
    The \emphi{maximum distance} between the two sets is
    \begin{equation*}
        \dmaxY{\PB}{\PC} =\max_{\pb \in \PB, \pc \in \PC} \dmY{\pb}{\pc}.
    \end{equation*}
\end{definition}

\begin{definition}
    For a metric space $\FMS$, and a point $c \in \FMS$, let
    \begin{equation*}
        \ball(c,r) = \Set{ p\in \FMS}{ \dmY{c}{p}\leq r }
    \end{equation*}
    denote the \emphi{ball} of radius $r$ centered at $c$.
\end{definition}

\begin{definition}
    \deflab{spread}%
    For a set $\P \subseteq \FMS$, its \emphi{closest pair distance} is
    \begin{equation*}
        \cpX{\P} = \min_{x,y \in \P: x \neq y} \dmY{x}{y}.
    \end{equation*}
    The \emphi{spread} of $\P$ is $\SpreadX{\P} = \diamX{\P} / \cpX{\P}$.
\end{definition}

\begin{definition}
    \deflab{o_times}
    For two sets $X,Y \subseteq \FMS$, let
    \begin{equation*}
        X \otimes Y =%
        \Set{ \{x, y\}}{x \in X, y \in Y, x \neq y}.
    \end{equation*}
    In particular, let $\binom{X}{2} = X \otimes X$ be the of all (unordered) pairs of $X$.
\end{definition}

\subsection{Different notions of separation}

These different notions of separation are illustrated in \figref{stability}.

\begin{definition}
    \deflab{well_separated}%
    A pair of sets $\{\PB, \PC\}$ of points in a metric space $(\FMS, \DistChar )$ is
    \begin{align*}
      &&&\text{\emphi{$\tfrac{1}{\eps}$-separated}}
      &\quad\text{if}\quad
      & \max \pth{\bigl. \diamX{\PB},
        \diamX{\PC} } \leq \eps \ts \dmY{\PB}{\PC}, \text{and}
        &
      \\
      &&& \text{\emphi{$\tfrac{1}{\eps}$-semi-separated}}
      &
        \quad\text{if}\quad{}
      &
      \min \pth{\bigl. \diamX{\PB},
        \diamX{\PC} } \leq \eps \ts \dmY{\PB}{\PC}.
        &
    \end{align*}
\end{definition}

A $\Eps$-well-separated pair is also $\Eps$-semi-separated.  See \figref{stability} for an example demonstrating these different concepts.

\begin{definition}
    \deflab{stable}%
    The \emphi{stability} of a pair $\pair = \{\PB, \PC \}$ is
    \begin{equation*}
        \fstabY{\PB}{\PC}
        =
        \frac{\dmaxY{\PB}{\PC} - \dminY{\PB}{\PC}}
        {2 \dminY{\PB}{\PC}}.
    \end{equation*}
    The pair $\pair$ is \emphi{$\eps$-stable} if $\fstabY{\PB}{\PC} \leq \eps$.
\end{definition}

\begin{remark}
    If $\pair = \{ \PB,\PC \}$ is $\eps$-stable, then
    \begin{math}
        \Bigl.\tfrac{\dmaxY{\PB}{\PC}}{\dminY{\PB}{\PC}}
        =
        1 + 2 \fstabY{\PB}{\PC}
        \leq
        1+2 \eps.
    \end{math}
    Namely, all pairwise distances in such a pair are roughly the same up to a factor of $1+2\eps$.  On the other hand, if the pair $\pair$ is $\Eps$-separated, then $\dmaxY{\PB}{\PC} \leq \diamX{\PB} + \diamX{\PC} + \dmY{\PB}{\PC} \leq (1+2\eps)\dmY{\PB}{\PC}$, and then $\fstabY{\PB}{\PC} \leq \eps$, implying that $\pair$ is $\eps$-stable.  Note, that an $\Eps$-semi-separated pair is not necessarily $\eps$-stable, see \figref{stability}.
\end{remark}

\begin{figure}
    \centering
    \begin{tabular}{|c|c|c|c|c|}
      \hline
      &
      {\includegraphics[scale=0.45,valign=c,angle=90]{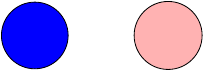}}
      &
        \includegraphics[page=2,scale=0.45,valign=c]{figs/pair}
      &
        \includegraphics[page=3,scale=0.225,valign=c]{figs/pair}
      &
        \includegraphics[page=5,scale=0.6,valign=c]{figs/pair}
      \\
      \hline
      $1$-well-separated
      &
        \Correct
      &
        \Wrong
      &
        \Wrong
      &
        \Wrong$\Bigr.$
      \\
      $1$-semi-separated
      &
        $\Correct$
      &
        $\Correct$
      &
        $\Correct$
      &
        \Wrong
      \\
      $1$-stable
      &
        $\Correct$
      &
        $\Correct$
      &
        \Wrong
      & \Correct
      \\
      \hline
    \end{tabular}
    \caption{%
       Rightmost example, if the sidelength of the grid is $1$, then $m = \dminY{\PB}{\PC} = 6$, $M = \dmaxY{\PB}{\PC} = 10$, and $\diamX{\PB} = \diamX{\PC} =\sqrt{37} \approx 6.08$. Thus, the pair $\{\PB, \PC\}$ is $\tfrac{1}{3}$-stable, as $\frac{M -m}{2m} = \tfrac{1}{3}$.}
    \figlab{stability}
\end{figure}

\subsection{Pairs decompositions}

\begin{definition}
    \deflab{pair_decomposition}%
    For a point set $\P \subseteq \FMS$, a \emphi{pair decomposition} of $\P$ is a set of pairs
    \[
        \WR = \brc{\bigl. \brc{\PB_1,\PC_1},\ldots,\brc{\PB_s,\PC_s}},
    \]
    such that
    \begin{compactenumI}[leftmargin=1cm]
        \item $\PB_i,\PC_i\subset \P$ for every $i$,
        \item $\PB_i \cap \PC_i = \emptyset$ for every $i$, and
        \item $\bigcup_{i=1}^s \PB_i \otimes \PC_i %
        = \binom{\P}{2}$.
    \end{compactenumI}
\end{definition}

\begin{definition}
    \deflab{WSPD}%
    For a point set $\P$, a \emphOnly{well-separated pair decomposition} of $\P$ with parameter $1/\eps$, denoted by \emphw{\Eps-\WSPD{}}, is a pair decomposition
    \begin{math}
        \WR = \brc{\bigl.  \brc{\PB_i,\PC_i}}_{i=1}^s
    \end{math}
    of $\P$, such that, for all $i$, the sets $\PB_i$ and $\PC_i$ are $\tfrac{1}{\eps}$-separated. The \emphi{size} of the \WSPD $\WR$ is the number of pairs in the decomposition (i.e., $s$).

    The \WSPD $\WR$ is a \pWSPD if every pair of $\binom{\P}{2}$ is covered exactly once by the pairs of $\WR$. Otherwise, it is a \cWSPD.
\end{definition}

\begin{definition}
    For a pair decomposition $\WR = \brc{\bigl.  \brc{\PB_i,\PC_i}}_{i=1}^s$ its \emphi{weight} is $\weightX{\WR} = \sum_i ( |\PB_i| + |\PC_i|)$.
\end{definition}
It is known that the weight of any pair decomposition of a set of $n$ points is $\Omega(n \log n)$ \cite[Lemma~3.31]{h-gaa-11}. In the worst case, the weight of a \WSPD of $n$ points is $\Omega(n^2)$.

\begin{definition}
    A pair decomposition $\WR$ of a point set $P$ is \emphi{semi-separated pair decomposition} (\emphi{$\Eps$-\SSPD{}}) of $P$ with parameter $1/\eps$, if all the pairs in $\WR$ are $\Eps$-semi-separated.
\end{definition}

Surprisingly, one can construct \eSSPD with total weight $O(\eps^{-d}n  \log n)$ with $O(\eps^{-d} n )$ pairs in $\Re^d$ \cite{adfgs-gswps-11}. Similar bounds are known for doubling metrics \cite{ah-ncsa-12}.

\subsection{Doubling metrics}
\seclab{doubling}

The \emphi{doubling constant} of a metric space $\FMS$ is the maximum, over all balls $\ball$ in the metric space $\FMS$, of the number of balls needed to cover $\ball$, using balls with half the radius of $\ball$. The logarithm of the doubling constant is the \emphi{doubling dimension} of the space.  The doubling dimension is a generalization of the Euclidean dimension, as $\Re^d$ has $\Theta(d)$ doubling dimension. Since we rarely know the precise value of the doubling constant, the value of the doubling dimension is usually not known precisely.

\subsection{Constructing \WSPD{}s via \SetCover}%
\seclab{set_cover}

Given a set $\P$ of $n$ points, one possible approach is to compute a set $\F$ of all maximal well-separated pairs in $\binom{\P}{2}$ i.e., if $\{ \PB, \PC \} \in \F$, then (i) $\PB, \PC \subseteq \P$, (ii) $\{\PB,\PC\}$ are well-separated, and (iii) there is no pair $\{\PB',\PC'\} \in \F$, such that $\PB \subseteq \PB'$ and $\PC \subseteq \PC'$. A pair $\pair = \{\PB,\PC\} \in \F$ covers the all the pairs of points in $\pairsX{\pair} = \PB \otimes \PC$.  Thus, the problem of finding a \minWSPD is equivalent to solving the minimum \SetCover problem for the instance $\bigl(\binom{\P}{2}, \Set{\pairsX{\pair}}{\pair \in \F} \bigr)$.

We can compute greedily an $\Theta(\log |\F|)$-approximation (in polynomial time) by adding at each step the pair in $\F$ that covers the maximum number of uncovered pairs in~$\Q$ \cite{s-tagas-97}. It is known that \SetCover cannot be efficiently approximated within a factor of $(1-o(1)) \ln n$, unless $\POLYT=\NP$ \cite{m-pgcnl-15}.  If one can limit the size of $\F$ to be polynomial in $n$, this readily leads to $O( \log n)$ approximation of $\minWSPD$.

Since $|\F|=O(n^2)$ for $d=1$ (see \secref{1D_WSPD_by_SetCover} for details), greedy \SetCover yields an $O( \log n)$-approximation for computing a minimum \eWSPD for a point set on the real line for any $\eps>0$.  For higher dimensions, no polynomial bound on $|\F|$ is known.

Naturally, one can restrict the sets used in building the pairs to be of a certain (smaller) family of sets. For example, in two dimensions, consider all the subsets of $\P$ of the form $\P \cap \square$ where $\square$ is some square. Then, one can solve/approximate the resulting \SetCover problem. However, even in two dimensions and squares, the resulting discrete \SetCover problem corresponds to a set of points in four dimensions, which one has to cover using a set of axis-aligned boxes. To the best of our knowledge, no better than $O( \log n)$ approximation is known in this case.

\section{Approximate biclique cover}
\seclab{abc}

\subsection{Preliminaries}

In the following, let  $\P$ be a given set of $n$ points in a metric space $(\FMS,\dC)$.

\begin{definition}
    For a parameter $\eps \in (0,1)$, a \defrefY{pair_decomposition}{pair decomposition} of $\P$ of the form
    \begin{math}
        \WR = \brc{\bigl.\brc{\PB_i,\PC_i}}_{i=1}^s
    \end{math}
    is \emphi{$\eps$-approximate biclique cover} (\emphi{\eABC}), if all its pairs are $\eps$-stable, see \defref{well_separated}.
\end{definition}
In other words, a \eABC is a biclique cover of the metric graph of $\P$, where for all $i$, the graph $(\PB_i \cup \PC_i, \PB_i \otimes \PC_i)$ is a biclique, where all the edges have roughly the same length, where an edge $\pb\pc$ is assigned the length $\dmY{x}{y}$.

\subsection{Constructing \ABC in the general metric case}

\begin{defn}
    A \emphi{ring} centered at $p$ with radii $r < R$ is the set $\ringZ{p}{r}{R} = \ballY{p}{R} \setminus \ballY{p}{r}$.
\end{defn}

\begin{theorem}
    \thmlab{abc_metric}%
    Given a set $\P$ of $n$ points in a metric space $(\FMS,\dC)$, and a parameter $\eps \in (0,1/2)$, one can construct a \eABC of $\P$ with $O( \tfrac{n}{\eps} \log n )$ pairs.
\end{theorem}
\begin{proof}
    The \ABC is constructed inductively. If $\cardin{P} = O(1)$ then $\WR = \binom{P}{2} = \Set{\bigl. \{ p, q\}}{p,q \in P}$.  Otherwise, let $p,q$ be the closest pair of points in $\P$, and let $Q = P \setminus \{p\}$. By induction, the point set~$Q$ admits a \eABC cover, say $\WR$. It remains to add the missing pairs~$\{p\} \otimes P$ to $\WR$. To this end, let
    \begin{equation*}
        M = 4+ \ceil{(32/\eps)\ln n}
    \end{equation*}
    and $\ell = \dmY{p}{q}$.  Let
    \begin{math}
        P_i = \ringZ{p}{r_i}{r_{i+1}} \cap \P,
    \end{math}
    where
    \begin{math}
        r_i = \ell (1+\eps/8)^{i-2},
    \end{math}
    for $i =1, \ldots, M$.  We add the pair $\pair_i = \bigl\{\{p\}, P_i \bigr\}$, for $i =1, \ldots, M$, to $\WR$.  These pairs cover all the pairs of points $x \in P$ such that $\ell \leq \dmY{p}{x} \leq n^4 \ell$.  Observe that the \defrefY{stable}{stability} of a pair $\pair_i$ is
    \begin{equation*}
        \fstabX{\pair_i}
        \leq%
        \frac{r_{i+1} - r_i }{ 2r_i}
        =
        \frac{(1+\eps/8)^{i-1} - (1+\eps/8)^{i-2}}{2(1+\eps/8)^{i-2}}
        =
        \frac{\eps}{16}.
    \end{equation*}

    For the remaining points $f \in \Pfar = P \setminus \ballY{p}{\ell n^4}$, consider the pair $qf$ (as a proxy for $pf$). It must be covered by an existing pair $\{\PB,\PC\} \in \WR$, where (say) $q \in \PB$ and $f \in \PC$. Add $p$ to the set $\PB$. This process is repeated for all the points in $\Pfar$.  Let $\WR'$ be the resulting pair decomposition of $P$, which we claim is the desired \eABC.

    Consider a pair $\pair_0 = \{\PB_0, \PC_0\}$ when it was first created. Initially,  $\pair_0$ is $\eps/16$-stable, by construction, and let $\fstabC_0$ denote its stability.  Let $\Delta_0 = \dmaxY{\PB_0}{\PC_0}$. Assume this pair evolved by the construction inserting points into it on either side, creating a sequence of pairs $\pair_0 = \{\PB_0,\PC_0\}, \ldots, \pair_t = \{\PB_t,\PC_t\}$, where $t \leq n$. Consider when the $i$\th point $p_i$ is added to the evolving pair $\pair_{i-1} = \{\PB_{i-1},\PC_{i-1}\}$ to create the new pair $\pair_i$, by (say) adding the point $p_i$ to $\PB_{i-1}$. This was done because there is a point $q_i \in \PB_{i-1}$ and $f_i \in \PC_{i-1}$, such that
    \begin{equation*}
        \dmY{p_i}{q_i} n^4
        \leq
        \dmY{q_i}{f_i}
        \leq
        \Delta_{i-1} = \dmaxY{\PB_{i-1}}{\PC_{i-1}}.
    \end{equation*}
    Thus, for $i=1,\ldots, t$, we have
    \begin{align*}
        \Delta_{i}
      &=
        \dmaxY{\PB_i}{\PC_i}
        \leq
        \Delta_{i-1} + \dmY{p_i}{q_i}
        \leq
        (1+1/n^4)
        \Delta_{i-1}
        \leq
        (1+1/n^{4})^n
        \Delta_{0}
        \\&%
        \leq
        (1+2/n^{3})
        \Delta_0.
    \end{align*}
    A similar argument implies that $\dmY{\PB_i}{\PC_i} \geq (1-4/n^{3}) \dmY{\PB_0}{\PC_0}$. A straightforward argument now shows that $\fstabY{\PB_t}{\PC_t} \leq \eps/8$, see \tedref{boring_1} (the argument uses that $n \geq 1/\eps$). Thus, all the pairs computed are $\eps/8$-stable, and the resulting decomposition is $\tfrac{1}{\eps}$-\ABC.

    We claim that the number of pairs generated is at most $f(n) = \frac c \eps n \ln n$ for some sufficiently large constant~$c$. By induction, the number of pairs in $\WR$ is at most $f(n-1)$, and the number of new pairs added to it to get $\WR'$ is $M$. For the resulting decomposition, the number of pairs is bounded by
    \[
        f(n-1) + M
        \leq%
        \frac c \eps (n-1) \ln (n-1) + 4+ \ceil{\tfrac{32}{\eps}\ln n}
        \leq%
        \frac c \eps n \ln n = f(n),
    \]
    for $c$ sufficiently large.
\end{proof}

\begin{remark*}
    We do not currently have a matching lower bound to \thmref{abc_metric}, and closing the gap is an interesting open problem.
\end{remark*}

\begin{remark}
    Observe that all the pairs of \thmref{abc_metric} start as ring sets, with one side being a single point. The sets in the pair get bigger by points being added to either side of this pair. The above argumentation implies that, at the end of the process, the diameter of the set that started as a singleton in the pair would still be at most $\eps$-fraction of the distance between the two pairs. Namely, the \ABC constructed in this case is also a $\tfrac{1}{\eps}$-\SSPD.

    Sadly, the resulting \ABC/\SSPD can have quadratic weight in the worst case. Indeed, consider a $1$-\ABC $\WR$ for the point set $P = \Set{6^i}{i =1,\ldots,n} \subseteq \Re$. If there is a pair $\{\PB, \PC\} \in \WR$, such that $\cardin{\PB} >1 $ and $\cardin{\PC} >1 $, then assume $\PC$ contain the largest number $6^\pc$ that appear in both sets, by disjointness of $\PB,\PC$, we have $6^\pc \notin \PB$.  Observe that the closest number to $6^\pc$, that is smaller than it, in $P$ is $6^{\pc -1}$. Thus, we have $\dmaxY{\PB}{\PC} \geq 6^\pc -6^{\pc-1} = 5 \cdot 6^{\pc-1}$.  Observe, that $\PC$ contains at least one other number $6^t$, for $t< \pc$. Thus, we have $\dminY{\PB}{\PC} \leq |\max(\PB) - 6^t| \leq 6^{\pc-1}$.  We thus have that
    \begin{equation*}
        \fstabY{\PB_i}{\PC_i}
        =
        \frac{\dmaxY{\PB}{\PC} - \dminY{\PB}{\PC}}
        {2 \dminY{\PB}{\PC}}
        \geq
        \frac{5\cdot 6^{\pc-1}  - 6^{\pc-1}}
        {2 \cdot 6^{\pc-1}}
        = 2.
    \end{equation*}
    Namely, such a pair can not be $1$-stable. Thus, all pairs in $\WR$ must have at least one of their sets as a singleton. But the weight of such a pair is the number of original point-pairs it covers (plus one). Since overall there are $\binom{n}{2}$ pairs, it follows that the weight of $\WR$ is at least $\binom{n}{2}$.
\end{remark}

\begin{corollary}
    \corlab{sspd_metric}%
    Given a set $\P$ of $n$ points in a metric space $(\FMS,\dC)$, and a parameter $\eps \in (0,1/2)$, one can construct a pair decomposition of $\P$ with $O( \tfrac{n}{\eps} \log n )$ pairs, that is both \eABC and \eSSPD.
\end{corollary}

\begin{observation}
    An easy inductive argument shows that the resulting \ABC{}s/\breakhere{}\SSPD{}s of \thmref{abc_metric} are disjoint -- every pair is covered exactly once by the cover computed.
\end{observation}

High-dimensional data often has a low spread (i.e., the spread is a small constant). By using only the rings in the construction of \thmref{abc_metric}, we readily get the following.

\begin{corollary}
    \corlab{abc_metric}%
    Given a set $\P$ of $n$ points in a metric space $(\FMS,\dC)$, with
    \defrefY{spread}{spread} $\Spread = \SpreadX{\P}$, one can construct a pair decomposition  of $\P$ of size $O\bigl(\tfrac{n}{\eps} \log \Spread\bigr)$, that is both  \eABC and \eSSPD.
\end{corollary}

\subsection{Construction \ABC in the Euclidean case}

\subsubsection{The construction}

We start with some standard definitions, which are taken from Chan \etal \cite{chj-lsoa-20}.  In the following, we fix a parameter $\epsA = 1/2^\LgEpsA$, for some integer $\LgEpsA >0$, and use
\begin{equation}
    \EpsA =2^\LgEpsA
    \qquad\text{and}\qquad %
    \LgEpsA = \log_2 \EpsA.
    \eqlab{eps}
\end{equation}

\begin{defn}
    \deflab{t_grid}%
    Let $\Cube \subseteq \Re^d$ be an axis-parallel cube with side length $\num$.  For an integer $t > 1$, partitioning $\Cube$ uniformly into a $t \times t \times \cdots \times t$ subcubes, forms a \emphi{$t$-grid} $\Grid(\Cube, t)$.  The diameter of a cube $\Cell$ is $\diamX{\Cell} = \mathrm{sidelength}(\Cell)\sqrt{d}$, and let $\cenX{\Cell}$ denote its \emphi{center}.
\end{defn}

\begin{defn}
    An \emphi{$\epsA$-quadtree} $\qte$ is a quadtree-like structure, built on a cube with side length $\num$, where each cell is partitioned into an $\EpsA$-grid.  The construction then continues recursively into each grid cell of interest.  As such, a node in this tree has up to ${1}/{{\epsA}^d}$ children, and a node at level $i \geq 0$ has an associated cube of side length $\num \epsA^{i}$. When $\EpsA = 2$, this is a regular quadtree.
\end{defn}

Informally, an $\epsA$-quadtree connects a node directly to its descendants that are $\LgEpsA$ levels below it in the regular quadtree. Thus, to cover the quadtree, we need $\LgEpsA$ different $\epsA$-quadtrees to cover all the levels of the quadtree. This cover can be realized by setting the level of the root of the $\epsA$-quadtree. Specifically, we fix the root of the $\epsA$-quadtree $\qte^i$ to be $[0,2^{i})^d$, for $i=1, \ldots, \LgEpsA$. We thus get the following.
\begin{lemma}
    \lemlab{reg_to_eps_quad}%
    Let $\qt$ be a regular (infinite) quadtree over $[0, 2)^d$.  There are $\LgEpsA$ $\epsA$-quadtrees $\qteA^1, \ldots, \qteA^{\LgEpsA}$, such that the collection of cells at each level in $\qt$ is contained in exactly one of these $\LgEpsA$ $\epsA$-quadtrees.
\end{lemma}

We need the following result on shifted quadtrees (see also \cite[Appendix~A.2]{chj-lsoa-20}).

\begin{lemma}[\cite{c-annqr-98},~Lemma~3.3]
    \lemlab{shifting}%
    Consider any two points $p, q$ in the unit cube $\UC$, and let $\qt$ be the infinite quadtree of $[0,2)^d$.  For $D= 2\ceil{d/2}$ and $i = 0, \ldots, D$, let $\shift_i = (i/(D+1), \ldots, i/(D+1))$. Then there exists an $i \in \brc{0, \ldots, D}$, such that $p + \shift_i$ and $q + \shift_i$ are contained in a cell of $\qt$ with side length $\leq 2(D+1) \dY{p}{q}$.%
\end{lemma}

\begin{observation}
    \obslab{ball_shift}%
    Consider a ball $\ballY{c}{\ell} \subseteq [0,1]^d$, and its axis aligned bounding cube $C$. Let $p$ and $q$ be the bottom and top corners of $C$.  Applying \lemref{shifting} to $p$ and $q$, we have that there is a node $v$ in a shifted quadtree that its sidelength is at most
    $L \leq  2(D+1) \dY{p}{q} \leq (2d + 4) \sqrt{d} 2\ell \leq 8 d^{3/2}\ell$, as $\dY{p}{q} = \sqrt{d} \cdot  2\ell$.
\end{observation}

\begin{theorem}
    \thmlab{abc_r_d}%
    For a set $\P$ of $n$ points in $\Re^d$, and $\eps \in (0,1)$, one can compute a pair decomposition $\WR$ of size $O( \tfrac{d^{3}n }{\eps} \log \tfrac{d}{\eps} )$. The decomposition  $\WR$ is \eABC and \eSSPD.
\end{theorem}

\begin{proof}
    Let $\epsA$ be the maximum number that is a power of $2$ and smaller than $ < \eps/(128d^{2})$. We construct the $D=2\ceil{d/2} \leq d+1$ shifted quadtrees of \lemref{shifting} over the points of $\P$ (we can assume $\P \subseteq [0,1]^d$ by shifting and scaling). Next, for each such shifted quadtree, we compute the $\LgEpsA$ $\epsA$-quadtrees each one of them induces. We compress the resulting $N = O( D \LgEpsA) = O(d \log \tfrac{d}{\eps})$ $\eps$-quadtrees -- they each store the $n$ points of $\P$ in their leaves, and each one of them has at most $n-1$ internal nodes.

    Now, given such a compressed $\epsA$-quadtree, for each of its nodes, $v$, we consider its $\EpsA$-grid $\Grid_v$, partitioning its cell $\Cell_v$, and assume the side length of cells in $\Grid_v$ is $\num$. For a cell $\Cell \in \Grid_v$ that contains at least one point of $\P$, we generate several pairs associated with $\Cell$ to the computed \ABC, as follows. Let
    \begin{equation*}
        \rho = \frac{8 \sqrt{d}}{\eps} \num.
    \end{equation*}
    For $i = 0, \ldots, 1/\epsA$, consider the spheres
    \begin{equation*}
        S_i
        =
        \Set{ x \in \Re^d }{ \dY{\cen }{x} =  \rho + i \num }.
    \end{equation*}
    Let $\Grid_{v,i}$ be the set of all the grid cells of $\Grid_v$ intersecting $S_i$, and let $\PC_i = \cup_{ C \in \Grid_{v,i}} (C\cap P)$, for $i = 0, \ldots, 1/\epsA$. We add the pairs $\pair_i = \{ \P \cap \Cell , \PC_i \}$ to the computed \ABC, for $i=0, \ldots, 1/\epsA$, see \figref{pairing}.  It is straightforward to verify that all these pairs are $\eps$-stable, see \tedref{boring_2}.

    \begin{figure}[ht]
        \centerline{%
           \includegraphics[width=0.5\linewidth]%
           {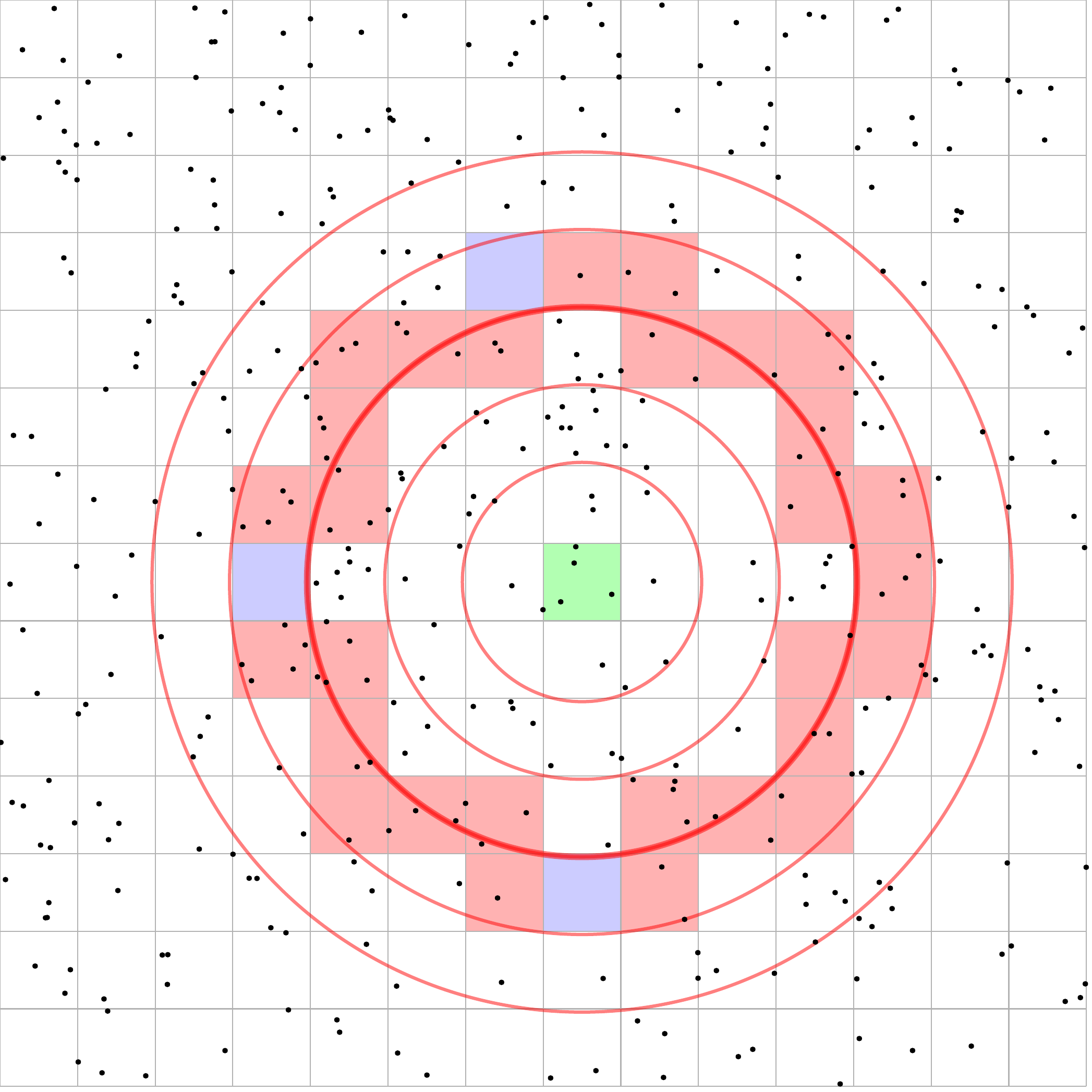}%
        }%
        \caption{The set of points in the green cell is paired with the set of points in the red cells (the blue cells are empty).}
        \figlab{pairing}
    \end{figure}

    Consider a specific compressed $\epsA$-quadtree $\qt$. A node $v$ in it gives rise to at most $O( n_v(1+1/\epsA) )$ pairs, where $n_v$ is the number of children $v$ has (i.e., $n_v = \cardin{\Set{C \in \Grid_{v}}{ C \cap P \neq \emptyset}}$). Clearly, $\sum_{v\in\qt} n_v = n-1$. Thus, each tree gives rise to $O(n/\epsA)$ pairs, and overall there are $O(n /\epsA) = O(d^{2} n/\eps)$ pairs generated for this tree. %
    Since there are $N$ $\epsA$-quadtrees, it follows that the total number of pairs is $O(d^{3} \tfrac{n}{\eps} \log \tfrac{d}{\eps} )$.

    As for the covering property, consider two points $p,q \in P$, with $\ell = \dY{p}{q}$. Let $\ball = \ballY{p}{ \ell }$. By \obsref{ball_shift}, there is a shifted quadtree that contains this ball in a cell with sidelength $L$, where
    \begin{equation*}
        2\ell
        \leq
        L \leq
        8 d^{3/2} \ell.
    \end{equation*}
    This cell $\Cell$ corresponds to a node $v$ in one of the computed $\epsA$-quadtrees. Let $\Cell_p$ and $\Cell_q$ be the two children cells of $\Cell$ that contains $p$ and $q$. The sidelength of $\Cell_p$ is $\num =\epsA L$.  Observe that $\eps/(256d^{2}) \leq \epsA \leq \eps/(128d^{2})$, and as such
    \begin{equation*}
        \frac{\ell}{\num}
        =%
        \frac{\ell}{\epsA L }
        \geq%
        \frac{256 d^{2} \ell}{\eps \cdot 8 d^{3/2} \ell }
        =%
        \frac{ 32 \sqrt{d} }{ \eps  }
        >%
        \frac{8 \sqrt{d} }{\eps } = \frac{\rho}{\num}.
    \end{equation*}
    This implies that $\ell > \rho$ for this cell. Similarly, we have
    \begin{equation*}
        \frac{\ell}{\num}
        =%
        \frac{\ell}{\epsA L }
        \leq%
        \frac{128 d^{2} \ell}{\eps \cdot 2  \ell }
        =%
        \frac{ 64 d^{2} }{ \eps  }
        \leq%
        \frac{1}{\epsA}.
    \end{equation*}
    Namely, the cell of $\Cell_p$ would be in the center of one of the generated sets in the pair decomposition, and it would contain the points of $P \cap \Cell_q$ on the other side.  We conclude that the pair $pq$ would be covered in the generated decomposition.
\end{proof}

\begin{remark}
    Every pair in the decomposition computed by \thmref{abc_r_d} can be represented implicitly (in a similar fashion to \WSPD). Indeed, each pair computed is a quadtree node coupled with a set of quadtree nodes, all children of the same $\epsA$-quadtree node, that intersect a certain sphere. The number of such grid cells in a set is $O(1/\epsA^{d-1} ) = O(1/\eps^{d-1})$. Namely, a single set can be represented as a list of $O(1/\eps^{d-1})$ nodes, where each node contributes all the points stored in its subtree to the set.
\end{remark}

\begin{remark}
    \thmref{abc_r_d} improves the number of pairs in the \eSSPD compared to previous work, from $O(n/\eps^d)$, to $O(n \cdot \tfrac{d^{3}}{\eps} \log \Eps)$. This result should improve some results using \SSPD{}s. However, the current construction of \SSPD{}s does not have the property that the weight of the \SSPD is near linear, and it can be quadratic in the worst case. It is quite conceivable that previous constructions can be modified to use the new ideas, and get a construction with an improved number of pairs and a small weight.

    Note that the resulting decomposition is not disjoint, and a single pair would get covered multiple times.
\end{remark}

\section{Constant approximation in doubling metrics}
\seclab{doubling_metrics}

\subsection{Preliminaries}
Let $P$ be a set of $n$ points in a metric space $(\FMS, \DistChar )$, with bounded doubling dimension $\dim$.  A standard construction of \WSPD{s} in doubling metrics is by Har-Peled and Mendel~\cite{hm-fcnld-06}. It is an adaptation of the construction in Euclidean spaces by Callahan and Kosaraju~\cite{ck-dmpsa-95}. The first step is constructing a hierarchical tree over the points, similar in spirit to quadtrees.

\begin{definition}
    \deflab{net_tree}%
    A \emphi{\nettree{}} of $P$ is a tree $T$ storing the points of  $P$ in its leaves. For a node $v\in T$, $P_v\subset P$ is the set of leaves in the subtree of $v$.  Each node $v$  has an associated representative point $\repX{v} \in P_v$.  Internal nodes have at least two children. Each vertex $v$ has a \emphi{level} $\mlevelX{v} \in \ZZ \cup\brc{-\infty}$.  The levels satisfy $\mlevelX{v}<\mlevelX{\parent(v)}$, where $\parent(v)$ is the parent of $v$ in $T$.  The levels of the leaves are $-\infty$.

    We require the following properties from $T$:
    \begin{compactenumI}
        \item \emphw{Covering}. For all $v\in T$:
        \begin{math}
            \ballY{\repX{v}}{\frac{2\tbase}{\tbase-1} \cdot \tbase^{\mlevelX{v}}} \supset P_v.\Bigr.
        \end{math}
        where $\tbase=11$.

        \item \emphw{Packing}. For all non-root
        $v\in T$:
        \begin{math}
            \ball\pth{\repX{v}, \smash{\frac{\tbase-5}{2(\tbase-1)}} \cdot \tbase^{\mlevelX{\parent(v)}-1}} \bigcap P \subset P_v.\Bigr.
        \end{math}

        \smallskip%
        \item \emphw{Inheritance}. For all nonleaf $u\in T$, there exists a child $v\in T$ of $u$ such that $\repX{u}=\repX{v}$.
    \end{compactenumI}
\end{definition}
\Nettrees are an abstraction of compressed quadtrees, as the levels of a node and its child can be dramatically different. Uncompressed quadtree corresponds to   \nettree where the levels of consecutive nodes along a path differ by exactly one (except the leaves). A \nettree where all the levels are consecutive is uncompressed.

\begin{theorem}[\cite{hm-fcnld-06}]
    \thmlab{net_tree}%
    Given a set $P$ of $n$ points in $\FMS$, one can construct a
    \nettree{} for $P$ in\/ $2^{\Odim}n \log n $ expected
    time, where $\dim$ is the doubling dimension of $\FMS$.
\end{theorem}

Once such a \nettree is available, computing the \WSPD is relatively straightforward.
\begin{lemma}[\cite{hm-fcnld-06}]
    \lemlab{wspd_d_m}%
    Let $P$ be a set of $n$ points in a metric space $\FMS$.  For $\eps \in (0,1)$, one can construct a \eWSPD of size $n \eps^{-\Odim}$, and expected construction time is $2^{\Odim}n \log n + n \eps^{-\Odim}\Bigr.$.  Furthermore, the pairs of the \WSPD are of the form $\{P_u,P_v\}$, where $u,v$ are vertices of a \nettree of $P$, such that $ \diamX{P_u}, \diamX{P_v} \leq \eps \dmY{\repX{u}}{\repX{v }}$.
\end{lemma}

\begin{proof}
    (Sketch.) The algorithm computes a \nettree{} $T$ using \thmref{net_tree}. Then, the algorithm maintains a FIFO queue of pairs of nodes $\{u,v\}$ that it has to separate. Initially, the queue consists of a single pair, that is, $\{ \rootX{T}, \rootX{T}\}$. The algorithm repeatedly considers the pair $\pair = \{u,v\}$ at the head of the queue, and checks if the pair is separated. For a pair $\pair$, define its \emphi{level} to be $\mlevelX{\pair} =\max( \mlevelX{u}, \mlevelX{v})$. The pair $\pair$ is defined to be separated if $8\tfrac{2\tbase}{\tbase-1}\cdot \tbase^{ \mlevelX{\pair}} \leq \eps \cdot \dmY{ \repX{u}}{ \repX{v}}.$ In other words, the distance between the representatives $\rho(u)$ and $\rho(v)$ is large enough to separate (the balls covering) $P_u$ and $P_v$. Now, if the pair $\pair$ is separated, then the pair $\{P_u, P_v\}$ is output. Otherwise, let $L(u) > L(v)$. The algorithm splits the bigger pair~$u$ by replacing $\pair$ in the queue with the pairs $C_u \otimes \{v\}$, where $C_u$ is the set of children of $v$.

    The output pairs have the form $\{P_u, P_v\}$. One can verify that every pair of points is covered by a well-separated pair $\{P_u, P_v\}$ in the output. Moreover, these pairs satisfy $\diamX{P_u}, \diamX{P_v} \leq \eps \dmY{\repX{u}}{\repX{v}}$, since by the covering property (I) of the net-tree, $\diamX{P_u}, \diamX{P_v} \leq \tfrac{\eps}{4} \dmY{P_u}{P_v}$, as $P_u \subseteq \ballY{\repX{u}}{r}$ and $P_v \subseteq \ballY{\repX{v}}{r}$, where $r = \frac{2\tbase}{\tbase-1} \cdot \tbase^{L}$.

    The standard charging argument for \WSPD{s} yields a size bound of $n/ \eps^{\Odim}$ pairs in the \WSPD. In particular, if $L(u) > L(v)$, we charge the pair $\{P_u, P_v\}$ to the parent of $v$ in the net-tree, and show using a packing argument in the doubling metric that each node $v \in T$ is charged at most $\eps^{-\Odim}$ times. For details, see \cite{hm-fcnld-06}.
\end{proof}

\subsection{Output-sensitive running time}

For the running time analysis, we improve on \cite{hm-fcnld-06} by providing an output-sensitive analysis. Let $\WR$ be the \eWSPD computed by the algorithm. Consider a history tree $H$ of the computation of the \WSPD, having the pairs being created by the algorithm as nodes. Every internal node in this tree has two or more children, and every leaf corresponds to a pair being output. Additionally, at most $O(n)$ superfluous leaves correspond to self-pairs $\{u,u\}$, with $u \in T$, and $u$ being a leaf. Overall, the size of $H$ is $N = O(n + \cardin{\WR})$. A careful inspection of the algorithm (see \lemref{wspd_d_m}) reveals that its overall running time is $2^{\Odim} n \log n + 2^{\Odim} \cardin{H} = 2^{\Odim} \pth{n \log n + \cardin{\WR}}$.  We state the output-sensitive running time as a corollary.

\begin{corollary}
    \lemlab{doubling_dimension_running_time}%
    The \WSPD algorithm in \lemref{wspd_d_m} runs in expected time $2^{\Odim} \pth{n \log n + \cardin{\WR}}$, where $\WR$ is the \eWSPD computed.
\end{corollary}

\subsection{Approximation factor for \minWSPD}

\begin{lemma}
    \lemlab{goody_goody}%
    The \WSPD algorithm of \lemref{wspd_d_m} outputs \eWSPD of size at most $2^{\Odim}\cardin{\Opt}$, where $\Opt$ is the optimal \eWSPD for $P$.
\end{lemma}

\begin{proof}
    \prflab{goody_goody}%
    Let $\WR$ be the \ \WSPD output by the algorithm.  It would be easier to argue here on the uncompressed \nettree. In this scenario, a pair $\{u,v\} \in \WR$ has the property that $\cardin{\mlevelX{u} - \mlevelX{v}} \leq 1$. Note that running the algorithm on the uncompressed \nettree, or the original \nettree, would output \emph{exactly} the same pairs. The difference is that in the compressed version, the algorithm can ``slide'' down a compressed edge in constant time, whereas the uncompressed version must laboriously go down the levels one by one.

    Consider an ``unbalanced'' pair $\{u,v\}$ with $\mlevelX{u} = \mlevelX{v} + 1$ (i.e., $u$ is larger than $v$). We replace this pair in $\WR$ by the pairs $C_u \otimes \{v\}$, where $C_u$ are the children $u$ in $T$. We repeatedly do this replacement till all pairs are balanced. This replacement process increases, by a factor of at most $2^{\Odim}$, the size of $\WR$, as this bounds the maximum degree in $T$. Let $\WR'$ be the resulting \eWSPD

    The set $\WR'$ is a list of pairs $\pair_1, \ldots, \pair_m$ of nodes of the \nettree.  And similarly, the optimal \eWSPD is a list of pairs $\opair_1, \ldots, \opair_k$, where $\opair_i = \{\PB_i,\PC_i\}$ for all $i$.

    For all $\pair = \{u,v\} \in \WR'$, we register $\pair$ with the (assume unique) optimal pair $\opair_j = \{\PB_j, \PC_j\}$ such that $\repX{u_i} \in \PB_j$ and $\repX{v_i} \in \PC_j$.

    Consider an optimal pair $\opair$, and let $\Delta = \diamX{ A \cup B}$.  Let $\pair_1, \ldots, \pair_t$ be all the pairs of $\WR'$ registered with $\opair$. All the pairs registered with $\opair$ must belong to levels $\mu, \mu-1$, for some integer $\mu$. Indeed, if not, consider the two extreme pairs $\pair,\pair'$ registered with $\opair$, with $L=\mlevelX{\pair}> \mlevelX{\pair'}+1$. But then, for $\pair' =\{u',v'\}$, a tedious but straightforward argument shows that $\{\parent(\repX{u'}), \parent(\repX{v'})\}$ would be considered to be separated by the algorithm.
    Implying that $\pair'$ has not been output by the algorithm, a contradiction.

    Let $Q$ be the set of all the representatives of pairs of level $\mu-1$, contained in $A$, where $\pair = \{A, B\}$. We have $\diamX{A} \leq \eps \Delta$, but by the packing property of the \nettree, the minimum distance between points of $Q$ is at least $\Omega( \eps \Delta )$. The doubling dimension implies that $\cardin{Q} \leq 2^{\Odim}$. Applying a similar analysis to $B$, and to registered pairs of level $\mu$, we conclude that the number of pairs of $\WR'$ registered with $\opair$ is $2 \cdot 2^{\Odim}\cdot 2^{\Odim} = 2^\Odim$.
\end{proof}

\subsection{The result}

\begin{theorem}
    Let $P$ be a set of $n$ points in a metric space $\FMS$ with doubling dimension $\dim$.  For $\eps \in (0,1)$, the algorithm of \lemref{wspd_d_m} constructs a \eWSPD of size $2^{\Odim} \cardin{\Opt}$, and in
    \begin{math}
       2^{\Odim}\pth{\cardin{\Opt}+n \log n}
    \end{math}
    expected time, where $\Opt$ is the minimal sized \eWSPD of $P$.
\end{theorem}

\section{Hardness of minimum-size \WSPD}
\seclab{hardness}

\subsection{Hardness of arbitrary metric spaces}
\seclab{hard_m_s}

The reduction is from \EdgeCliqueCover.
\begin{problem}[\EdgeCliqueCover]
    An instance is a pair $(\G,k)$, where $\G$ is a graph, and one has to decide if one can cover the edges of $\G$ by $k$ cliques that are contained in $\G$.
\end{problem}

\EdgeCliqueCover is \NPHard, and it is \NPHard to approximate within a factor of $n^{1-\delta}$, for some $\delta \in (0,1)$, see \cite{ly-hamp-94}.

\begin{theorem}
    \thmlab{NP_hard_arbitrary}%
    Computing the (exact) minimum size $3$-\WSPD is \NPHard on an arbitrary metric space.
\end{theorem}

\begin{proof}
    Consider an instance $(\G,k)$ of the \EdgeCliqueCover problem, with $n$ vertices.  Construct a metric space $(\FMS,\dC)$ with $\VV = \VX{\G}$ as the initial set of elements.  Let $\dmY{u}{v} = 1$ if $uv \in \EGX{\G}$, and $2$ otherwise. The metric space $(\FMS, \DistChar)$ is a $1/2$ metric space.  Add a new point $p$ to $\FMS$, such that $\dmY{p}{v} = 3$ for all $v \in \FMS$.  The finite metric space $\FMS$ has a $3$-\WSPD of size $\binom{n}{2} +k$ if and only if \CliqueCover on $G$ has a clique cover of size $k$.

    Indeed, if $\G$ has a clique cover $C_1, \ldots, C_k$, then a corresponding $3$-\WSPD, of the state size, is
    \begin{equation*}
        \WR%
        =%
        \bigl\{ \{p\}, \VX{ C_i}\bigr\}_{i=1}^k %
        \cup%
        \binom{\VV}{2}.
    \end{equation*}
    The set $\WR$ is a valid $3$-\WSPD for $\FMS$, as it covers all the pairwise points in $\FMS$, and the diameter of each $C_i$ is one by construction, and thus, all the pairs constructed are $3$-separated.

    Conversely, let $\WR$ be a $3$-\WSPD of size $N$.  Any two elements $u, v \in \VV$ have distance either $1$ or $2$.  Thus, as $\eps = 1/3$, all the pairs involving only points of $\VV$ must be singletons, as otherwise they would not be $3$-separated. Similarly, the ``special'' point $p$ must be a singleton in any pair of $\WR$ including it, as any set containing $p$ and some other point of $\VV$ has diameter $3$, which is not $3$-separable from any set of points in $\FMS$. Thus, the pairs in $\WR$ can be broken into two types of pairs:
    \begin{compactenumI}
        \smallskip%
        \item Singleton pairs $\bigl\{\{x,y\}\bigr\}$, with $x,y \in \VV$. There are exactly $\binom{n}{2}$ such pairs.

        \smallskip%
        \item Pairs $\pair_i = \bigl\{ \{p\}, \VV_i \}$, for $i=1,\ldots, t$,
        where $N = \binom{n}{2} + t$ and $\VV_i \subseteq \VX{\G}$.
    \end{compactenumI}
    \smallskip%
    If the induced subgraph $\G_{\VV_i}$ form a clique, then $\diamX{\VV_i} =1$, and the pair $\pair_i$ is $3$-separated. If $\VV_i$ is not a clique then $\diamX{\VV_i} = 2$, and then $\pair_i$ is not $3$-separated. We conclude that $\WR$ encodes a $N - \binom{n}{2}$ clique cover for $\G$.

    Thus, $\FMS$ has a $3$-\WSPD  of size $\binom{n}{2} + k$ $\iff$ $\G$ has a clique cover of size $k$.
\end{proof}

\begin{lemma}
    \lemlab{hard_m_s}
    For a finite metric space $\FMS$ defined over a set of $n$ points, for $\delta \in (0,1)$ sufficiently small, \emph{no} $n^{\delta}$-approximation to the size of the optimal $3$-\WSPD of $\FMS$ is possible, unless $\POLYT=\NP$.
\end{lemma}
\begin{proof}
    Let us repeat the above proof, but use $\beta \gg n^2$ copies of the original graph. Given a graph $\G = (\VV,\Edges)$, with $n = \cardin{\VV}$ and $m = \cardin{\Edges}$, consider the set $U = \VV \times \IRX{\beta}$, where $\beta$ is an integer and $\IRX{\beta} = \{1, \ldots, \beta\}$. Let $\FMS=(U,\dC)$ be a metric space, where for two points $(x,i), (y,j) \in U$, we define their distance to be (i) $3$ if $i \neq j$, (i)  $1$ if $xy \in \Edges$, and (iii)  $2$ if $i = j$ and $xy \notin \Edges$. Intuitively, we have $\beta$ copies of the $1/2$ metric from the previous proof, with distance $3$ between the copies.

    If $\G$ has an (edge) clique cover of size $k$, with cliques in the cover having the vertices sets $C_1, \ldots, C_k$, then setting $C_{\alpha,i} =  C_\alpha \times \{i\}$, we have the $3$-\WSPD:
    \begin{align*}
      \WR
      &=
        \Set{\Bigl. \{ C_{\alpha,i}, C_{\beta,j}\} }{ i,j \in \IRX{\beta}, \alpha,\beta \in \IRX{k}, i< j}
      \\
      &\qquad \cup \Set{\Bigl. \bigl\{\{(u,i)\}, \{(v,i)\}\bigr\} }{ i \in \IRX{\beta}, u,v \in \IRX{n}, u \neq v}.
    \end{align*}
    The first part matches all the cliques between different layers, and the second part covers each layer's vertices with $\binom{n}{2}$ singleton pairs. Overall, this \WSPD has size
    \begin{equation*}
        \xi = \binom{\beta}{2} k^2 + \beta \binom{n}{2}.
    \end{equation*}

    Assume that we had somehow computed a $3$-\WSPD $\WR$ of $\FMS$ with $\alpha \xi$ pairs, for some $\alpha < k$ (here, the clique cover is not provided). A pair $\{X, Y\} \in \WR$, must have the property that there are indices $i, j$, such that $X \subseteq \VV \times \{i\}$ and $Y \subseteq \VV \times \{j\}$, as otherwise $\diamX{X} \geq 3$ or $\diamX{Y} \geq 3$, which in either case would imply the pair can not be $3$-separated as $\diamX{\FMS}=3$.

    Let $\WR_{i,j}$ be all the pairs in $\WR$ between subsets of $\VV\times\{i\}$ and $\VV\times \{j\}$. By averaging, there are indices $i,j$ such that $\cardin{\WR_{i,j}} \leq \cardin{\WR}/\binom{\beta}{2} \leq \alpha k^2 + 1 \leq k^3$. Interpreting all the sets of $\WR_{i,j}$ covering $\VV \times \{i\}$ as an edge clique-cover of $\G$, implies that we had computed a $k^3$ cover of $\G$. To see why this is impossible, assume $k = O(n^{\delta/8})$, for some sufficiently small constant $\delta \in (0,1)$ (this can be ensured by adding a sufficiently large clique to the original graph $\G$). We thus had computed an $k^2$-approximate clique cover for $\G$, where $k^3/k = O(n^{\delta/4}) \ll O(n^{1-\delta})$. However, it is known \cite{ly-hamp-94} that unless $\POLYT = \NP$, no approximation of quality $n^{1-\delta}$ is possible for edge clique cover.
\end{proof}

\subsection{Hardness in 2d}
\subsection{\NPHardness{} in the Euclidean plane}
\apndlab{2_d_hardness}

\subsubsection{Background}
\seclab{p_c_background}

Given a metric space $(\FMS, \dC)$, a set of points $\P \subseteq \FMS$, and a set of centers $\CS \subseteq \FMS$, the $k$-center clustering \emphw{price} of $\P$ by $\CS$ is
\begin{math}
  \KCenPrcY{\P}{\CS}%
  =%
  \max_{p \in P} \dmY{p}{\CS}%
  =%
  \max_{p \in P} \min_{c \in \CS} \dmY{p}{c}.
\end{math}
Every point in $\P$ is within distance at most $\KCenPrcY{\P}{\CS}$ from its nearest center in $\CS$.

\begin{definition}
    The \emphi{$k$-center problem} is to find a set $\CS \subseteq \FMS$ of $k$ points, such that $\KCenPrcY{\CS}{\P}$ is minimized:
    \begin{math}
      \KCenPrcOptY{\P}{k}%
      =%
      \min_{\CS \subseteq \P,
      \cardin{\CS} = k} \KCenPrcY{\P}{\CS}.
    \end{math}
\end{definition}

\begin{definition}[\PairwiseClustering]
    The \emphi{$k$-pairwise clustering} problem is to find a partition of $\P$ into $k$ disjoint sets $\Pi = \{P_1, \ldots, P_k \}$, such that the \defrefY{metric_space}{diameter} $\diamX{\Pi} = \max_{X \in \Pi} \diamX{X}$ is minimized. The optimal such diameter is denoted by
    \begin{equation*}
        \optCDiamY{k}{\P} = \min_{\Pi \in \IRX{k}^\P} \diamX{\Pi},
    \end{equation*}
    where $\IRX{k}^\P$ denotes the set of all $k$ disjoint partitions of $\P$.
\end{definition}

Gonzalez \cite{g-cmmid-85} presented a simple, $O(kn)$ time, $2$-approximation algorithm to the optimal $k$-center clustering radius, and it is not hard to show that no better approximation exists in general metric spaces. The same algorithm provides a $2$-approximation for the \PairwiseClustering problem. Feder and Greene \cite{fg-oaac-88} showed that in two dimensions it is $\approx 1.822$ hard to approximate the optimal $k$-center radius. For \PairwiseClustering, they showed a slightly stronger hardness.

\begin{theorem} \hcite[Theorem 2.1]{fg-oaac-88} %
    \thmlab{f_g_pairwise}%
    \PairwiseClustering, in the plane (under the $L_2$-norm) is, \NPHard to approximate within a factor of $\alpha = 2\cos 10^\circ  \approx 1.969$.
\end{theorem}

The proof of Feder and Greene \cite{fg-oaac-88} works via a reduction from \VertexCover for planar graphs of degree at most $3$. Specifically, given such a graph, they carefully replace its edges with long chains of points, and the vertices are carefully designed to merge such chains. In the resulting point set $P$ (which looks ``like'' an embedding of the original graph), there is a parameter $k$, such that either $P$ can be covered by $k$ clusters of diameter $1$, in which case the original graph has a vertex-cover of size $k$. Alternatively, any cover of the resulting point set by $k$ clusters has \PairwiseClustering diameter $\alpha$ ($k$ here is quite large).

\subsubsection{Reduction}
\seclab{reduction}

\begin{theorem}
    \thmlab{NP_hard_euclidean_2}%
    Computing the minimum size \WSPD is \NPHard for point sets in the Euclidean plane. Specifically, given a point set $Q$ in the plane, and parameters $t \in \NN$ and $\eps \in (0,1)$, such that deciding if $Q$ has a \eWSPD of size at most $t$ is \NPHard.
\end{theorem}

\begin{proof}
    Let $(P,k)$ be an instance of \PairwiseClustering problem in the plane as generated by the reduction of \thmref{f_g_pairwise}.  The task at hand is to decide if $\P$ can be partitioned into $k$ clusters of diameter $1$, or $\P$ has the property that any such partition has \PairwiseClustering diameter $\alpha$.  Let $n = \cardin{P}$, and observe that by the construction of $P$, we have $\diamX{P} \leq 2n$. Note, that for $\eps < 1/4n$, any \eWSPD of $\P$ is an exhaustive list of all $\binom{n}{2}$ pairwise singletons.

    The idea is to add a point $p$ to $P$ that is far away. Let $p$ be a point in distance (say) $\ell = n/\delta$ from its closest point $p' \in P$, where $\delta \in (0,1)$ is a sufficiently small constant to be specified shortly. For any $X \subseteq P$, we have that its separation quality in the pair
    $\bigl\{ \{p\}, X\bigr\}$ is
    $\frac{\diamX{X}}{\dmY{p}{X}}$, which is bounded in the range
    \begin{equation*}
        (1-2\delta) \frac{\delta}{n}
        \diamX{X}
        \leq
        \frac{\diamX{X}}{(2 + 1/\delta) n}
        \leq
        \frac{\diamX{X}}{\ell + 2n}
        \leq
        \frac{\diamX{X}}{\dmY{p}{X}}
        \leq
        \frac{\diamX{X}}{\ell}
        =
        \frac{\delta}{n}\diamX{X},
    \end{equation*}
    since
    \begin{math}
        \frac{1}{2 + 1/\delta}
        =%
        \frac{\delta}{1+ 2\delta}
        \geq
        \delta(1-2\delta).
    \end{math}

    Now, the gap in the \PairwiseClustering diameter of $\P$ is between $1$ and $\alpha \approx 1.969$. For $\delta = \min(1/4n, 10^{-5})$, a set $X \subseteq P$ of diameter $1$, the corresponding pair $\bigl\{ \{p\}, X\bigr\}$ has separation in the range $0.999 \frac{\delta}{n}$ to $\tau = \frac{\delta}{n}$. On the other hand, if $X$ has diameter $\alpha$, then this pair has separation at least $0.999 \frac{\delta}{n} \alpha > 1.96 \tau$.

    So, consider a \eWSPD $\WR$ of $Q = \P \cup \{ p \}$, for $\eps = \tau = \min(1/4n, 10^{-5})/n$. This \WSPD has $\binom{n}{2}$ singleton pairs to handle the points in $P$. All the pairs involving $p$ must have $p$ as a singleton on one side of the pair. The other side can have a diameter of at most $1$ if we demand that it be $1 / \tau$-separated. If $X$ has diameter $\alpha$ or bigger, than the pair would not be $1 / \tau$-separated. It follows, that if $\P$ has a $k$ \PairwiseClustering with $k$-clusters of diameter $1$, then the corresponding \eWSPD would be of size $\binom{n}{2} + k$. The implication also works in the other direction. If $\P$ has a minimum $k$-\PairwiseClustering diameter $\alpha$, then any $\psi$-\WSPD of $Q$ of size $\binom{n}{2} + k$, must have separation at least (say) $(0.99 /\alpha)/\tau$.
\end{proof}

\section{The one-dimensional problem}
\seclab{1_dim}

\subsection{Approximating \minWSPD via  \SetCover}
\seclab{1D_WSPD_by_SetCover}

Let $\P = \{ p_1, \ldots, p_n \} $ be a set of $n$ points in the real line $\Re$, where $p_i < p_{i+1}$ for all $i$.  As described in the introduction and in \secref{set_cover}, we can reduce \cWSPD to \SetCover.  We aim to cover the set $\Q = \Set{ (p_i,p_j)}{ i < j}$ of all $\binom{n}{2}$ grid points strictly above the diagonal in the non-uniform grid $\P \times \P$.

A (maximal) well-separated pair for $\P$, is a pair $\pair = \{ I, J\}$, where $I, J$ are intervals on the line, and $I$ and $J$ are $\tfrac{1}{\eps}$-well-separated. Formally, we require that $\lenX{I}, \lenX{J} \leq \eps \dmY{I}{J}$, where $\lenX{I}$ denotes the length of $I$.  We might as well pick $I$ and $J$ to be maximal while preserving the separation property. That is $\lenX{I} = \lenX{J} = \eps \dmY{I}{J}$. In addition, if the two endpoints of $I$ and $J$ closest to each other are not points of $\P$, we can move the two intervals apart, increasing their separation.  Thus, for a fixed $\eps > 0$, all the (maximal) separated pairs $\{I, J\} \in \F$ are uniquely defined by their two inner endpoints, which belong to $\P$.  Therefore, the size of $\F$ is $\binom{n}{2}$.

When we order the pair $(I, J)$, denoting that $I < J$ (i.e., for all $x \in I$ and $y \in J$, we have $x< y$), the two inner points $p_i,p_j$ form a point $u=(p_i,p_j)$ in $Q$. This gives a unique relation between $I \times J$ and the square $\shadowX{u} := [p_i-\tau, p_i] \times [p_j,p_j+\tau] $, where $\tau = \eps (p_j - p_i)$, by the bottom-right point $u$. The square $ \shadowX{u}$ is the \emphi{shadow} of $u$, and $u$ is the \emphi{anchor} of $ \shadowX{u}$.

Thus, the \minWSPD of $\P$ corresponds to a minimal set $\SQSet \subseteq \Set{ \shadowX{u}}{ u \in Q}$ that covers the points of $\Q$. This problem can be $(1+\delta)$-approximated in $n^{O(1/\delta^2)}$ time using local search \cite{dl-gdssc-23,mr-pghsp-09}, as this is an instance of discrete set cover by pseudodisks. Thus, we get the following.

\begin{lemma}
    \lemlab{1_d_cover}%
    Given a set $\P$ of $n$ points in $\Re$, and parameters $\eps,\delta \in (0,1)$, one can compute, in $n^{O(1/\delta^2)}$ time, a \ecWSPD of $\P$ with $k$ pairs, where $k \leq (1+\delta)\Opt$, where $\Opt$ is the minimum number of pairs in any \ecWSPD of $\P$.
\end{lemma}

\subsection{A 3-Approximation for \cWSPD}
\seclab{1_dim_3_approx}

As in \secref{1D_WSPD_by_SetCover}, we interpret $Q$ as the set of all points above the diagonal in the non-uniform grid $\P \times \P$, and search for a set of squares that covers $Q$. The algorithm uses the ``One for the Price of Two'' approach of Bar-Yehuda \cite{b-optua-00}. We compute a set of three squares that provides at least as much coverage as a single square in the optimal solution.

\begin{definition}
    \deflab{vicinity}%
    The \emphi{vicinity} $V(p)$ of a point $p\in \Q$ is the set of all points $p'\in \Q$, such that $p,p'\in \shadowX{r}$ for $r \in \Re^2$. That is, the vicinity of $p$ is the region of all the points in the plane that might be covered by a single shadow square that also covers $p$.
\end{definition}
See \figref{cover_construction} for an example of a vicinity.
We can reduce to $r \in Q$, because if there exists a $\shadowX{r}$ with $p, p' \in \shadowX{r}$ than $r' := (\min(p_i, p'_i), \min(p_j, p'_j)) \in Q$ is an anchor of a shadow with $p, p' \in \shadowX{r'}$.

Now, consider the optimal solution $\OPT \subseteq \F$, and consider a single point $p \in Q$ not yet covered. The point $p$ is covered by $\OPT$, and for all squares $\square \in \OPT$ that contain $p$, we have $\square \subset V(p)$ by definition.
Our main idea is to sweep the points $Q$ from right to left and cover the left half of the vicinity of a point with three squares (compare to \figref{cover_construction}).

\begin{figure}[ht]
    \phantom{}%
    \hfill%
    \includegraphics[page=1]{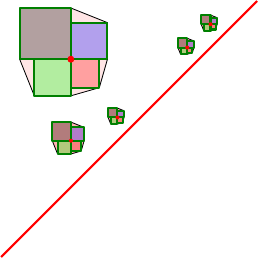}
    \hfill%
    \includegraphics[page=2]{figs/vicinity}
    \hfill%
    \phantom{}  \caption{Left: A few points and their vicinities, for $\eps =0.4$.  Right: The left half of the vicinity of a point $p$ can be covered by three squares.}%
    \figlab{cover_construction}
\end{figure}

\sparagraph{Algorithm.}
We sort the input $P$ to sweep our points $Q$ from right to left. As our points have a grid structure, we can process them column-wise.  We only process uncovered points and add squares to cover the left half of the vicinity of the current point.

Horizontal and vertical intervals define a square.  The horizontal interval $I$ determines when an interval is inserted or deleted during the sweep, i.e.\ for which columns the square is relevant.  The vertical intervals $J$ cover points in the current column.  We use a segment tree as status for the latter, see Bentley's algorithm~\cite{lw-mprrds-81}. A leaf of this segment tree is a point $p_i \in P$ representing the point $(p_i,p_j)$, where $p_j$ is fixed by the current column. We can then query the segment tree for the next uncovered point, ordered from bottom to top, in this column. To ensure that no grid points below the diagonal are considered, for each column, we add the interval below the diagonal into the segment tree.

We have two events: Processing an uncovered point and starting a new column. When processing an uncovered point $p$, we add three squares to cover the left half of its vicinity (see \figref{cover_construction}):
\begin{itemize}
    \item $\shadowX{p}$ (i.e., the shadow of $p$), %
    \item $S$, which is the unique shadow square that has $p$ as its top right corner, and  %
    \item $\shadowX{q}$, where $q$ is the bottom left corner of $S$.
\end{itemize}
Note that $S$ and $\shadowX{q}$ might not be in $\F$. However, as described above, we can move their anchor up resp.\ up and left to a point in $\Q \cap V(p)$. We only add $S$ and $\shadowX{q}$ if they contain a point from $V(p)$ besides $p$.

We insert the vertical segments of $\shadowX{p}$ and $S$ into the segment tree.
Since the anchor of $\shadowX{q}$ is left of the current column, we insert $\shadowX{q}$ into a priority queue, prioritized by the column. In another priority queue, we add the delete events for the three squares, each with a priority indicating when to be deleted from the segment tree (defined by the horizontal segments).

When a new column is started, one checks the priority queues for squares whose intervals need to be added or deleted to the segment tree.
Furthermore, at the start of a new column, one adds an interval to the segment tree to cover everything in the current column below the diagonal, and deletes the corresponding interval of the previous column.

We proceed with the following column if the segment tree contains no uncovered points. A visualization of the sweep and event status is shown in \figref{1D-sweep-status}.

\begin{figure}[ht]
    \phantom{}\hfill%
    \begin{minipage}{0.34\linewidth}
        \includegraphics[page=1,,width=0.99\linewidth]{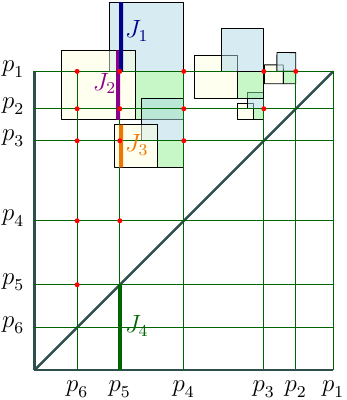}
    \end{minipage}
    \hfill%
    \begin{minipage}{0.45\linewidth}
        \includegraphics[page=2,width=0.99\linewidth]{figs/sweep_seg_tree}
    \end{minipage}%
    \hfill\phantom{}%

    \caption{Visualization of the sweep and event status at the beginning of column $p_5$. A search would find $p_4$. When we place an interval over it on the backtrack, we check if the children of a node are covered; if they are, we mark the node as covered. So, covering $p_4$ propagates to the root node, and the next search would terminate at the root.}
    \figlab{1D-sweep-status}
\end{figure}

\sparagraph{Approximation factor.}
Since we cover the vicinity of a point $V(p)$, our solution covers the optimal square(s) that cover $p$. Therefore, we get $3$ as the approximation factor, which is the minimum number of squares to cover the vicinity of a point in any sweep direction.

\begin{lemma}
    For any $\eps>0$, only three squares are needed.
\end{lemma}
\begin{proof}
    To account for every sweep's direction, consider all halfplanes through $p$. First, observe that one can cover the two left quadrants with three squares. This argument gives us the lower bound.  We highlight five features also shown in \figref{features}. The upper left point $r$, the neighborhood of $p$, the diagonals in the upper right and lower left quadrant, and the two diagonals (as a single feature) in the lower right quadrant. All halfplanes, except for three, contain an open neighborhood of at least four of them. Every square that contains such an open neighborhood can not cover any other. One direction is your solution, another is a symmetrical version. The last one contains only two features. However, it contains the lower right quadrant, which can not be covered with three squares. So we need at least $3$ squares, which gives a tight upper bound.
    \begin{figure}
        \phantom{}%
        \hfill%
        \includegraphics[page=1]{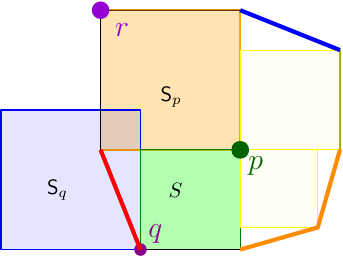}
        \hfill
        \includegraphics[page=2]{figs/small_proof}
        \hfill
        \phantom{}%

        \caption{The five features in different colors and an angle plot with intervals show which direction includes which features. The three arrows highlight directions of the normal vector, which are the only ones with fewer than four features.}
        \figlab{features}
    \end{figure}
\end{proof}

Finally, we analyze the running time.

\begin{lemma}
    \lemlab{runtime}%
     Given a set $\P$ of $n$ points in $\Re$, and parameter $\eps>0$, one can compute, in $O(n \log n)$ time, a cover \eWSPD of $\P$ with $k$ pairs, where $k \leq 3\Opt$, where $\Opt$ is the minimum number of pairs in any cover \eWSPD of $\P$.
\end{lemma}
\begin{proof}
    The runtime analysis relies on the optimal solution having $O(n)$ many squares.  A square can be computed in constant time.  We only insert and delete intervals for each column and the uncovered point. We can charge the uncovered points on the optimal intervals to get $O(n)$ many. Each insertion and deletion is in $O(\log n)$ time. Computing the subsequent uncovered point takes $O(\log n)$ time. For this purpose, we maintain a flag for each tree node indicating whether it is fully covered. This tracking can be performed while inserting and deleting intervals, and verified by checking whether both children are fully covered during backtracking. Maintaining the priority queue takes $O(n \log n)$ overall. Each square introduces only one insertion, deletion, and a point lookup, and the overall runtime is $O(n \log n)$.
\end{proof}

\subsection{From a \cWSPD to a \pWSPD}
\seclab{partition}

If one requires the separated pairs to be disjoint, we give a constant approximation.

\begin{lemma}
    \lemlab{squares_partition_2}%
    Let $\SQ$ be a set of $n$ axis-aligned squares in the plane. Then one can compute, in polynomial time, a set $\R$ of at most $9n$ interior-disjoint rectangles, such that $\cup_{\square \in \SQ} = \cup \SQ = \cup \R$. Furthermore, for each $\rect \in \R$, there exists a square $\square \in \SQ$, such that $\rect \subseteq \square$.
\end{lemma}
\begin{proof}
    Let $\SQ = \set{\sq_1, \ldots, \sq_n}$ be this set of squares numbered from largest to smallest.  We next insert them into the union one by one.  Imagine collecting all the vertices that appear on the boundary of the union during this process into a set $U$.

    In the $i$\th iteration, at most $8$ new vertices are added to $U$.  Indeed, an edge of the smallest square (i.e., $\sq_i$) can support at most two vertices in the union.  Indeed, as one traverses an edge $e$ of $\sq_i$, once it enters into a larger square $\sq_j$, for $j < i$, it can not reappear on the union, as its endpoint lies in the interior of $\sq_j$. Thus, the portion of $e$ that appears on the boundary of $\cup_{t=1}^i \sq_i$ is either empty or a single segment. Thus, $\sq_i$ can contribute at most four edges to the union of squares, and hence at most $8$ vertices to $U$.

    For $i=1,\ldots,n$, let $\UU_i = \Re^2 \setminus \cup_{t=1}^i \sq_i$ be the complement of the union of the first $i$ squares of $\SQ$, and let $R_i = \UU_i \cap \sq_i$ be the region being newly covered in the $i$\th iteration. The algorithm in the $i$\th iteration adds the rectangles forming the vertical decomposition of $R_i$ to the output set.

    For $i=1,\ldots,n$, let $V_i = \VX{\partial{\UU_i}}$ be the set of vertices of the union, and let $B_i = V_{i-1} \cap \sq_i$. The vertices of $B_i$ are being ``buried'' in the interior of the union in the $i$\th iteration.

    Let $b_i = \cardin{B_i}$. If $b_i = 0$, then $\sq_i$ is an ``island'', and a new rectangle is added to the output. More generally, a somewhat naive upper bound on the number of rectangles needed in forming the vertical decomposition of  $R_i$ is $b_i +1$. Indeed, imagine sweeping $R_i$ from left to right. Every time a new trapezoid (except the first one) is created, it is because of a newly encountered vertex of $B_i$.

    Since every vertex being covered was created by some earlier iteration, we have $\sum_{i=1}^n b_i \leq 8m -4$. Thus, the number of rectangles, output by the algorithm, is bounded by $9n$.
\end{proof}

\begin{remark}
    A similar result to \lemref{squares_partition_2} holds for any (convex) pseudo-disks, but the constant is somewhat worse. A direct proof follows by randomized incremental construction drawing the objects from ``back to front'', see \cite[Section 2]{hks-akltd-17}.
\end{remark}

Combining \lemref{1_d_cover} and \lemref{squares_partition_2} implies the following result.

\begin{lemma}
    \lemlab{1_d_partition_2}%
    Given a set $\P$ of $n$ points in $\Re$, and a parameter $\eps \in (0,1)$, one can compute, in $n^{O(1)}$ time, a partition \eWSPD of $\P$ with $k$ pairs, where $k \leq 10 \Opt$, where $\Opt$ is the minimum number of pairs in any cover (or partition) \eWSPD of $\P$.
\end{lemma}

\section{Experiments}
\seclab{experiments}

We implemented our algorithms for the one-dimensional case.  The implementation was done using \Julia, and \Gurobi to solve the associated integer program \IP. Specifically, we implemented the following:
\begin{compactenumI}[leftmargin=1cm]
    \smallskip%
    \item Greedy: Lifts the point set from one dimension to the 2d grid, and computes their associated anchored squares induced by the points. Then uses the standard greedy \SetCover algorithm that repeatedly picks the set that covers the most uncovered elements.

    \smallskip%
    \item \WSPD: Uses the \WSPD algorithm as described by Callahan and Kosaraju \cite{ck-dmpsa-95}. This algorithm outputs a disjoint cover, and it is surprisingly efficient in our experiments.

    \smallskip%
    \item \AprxT: The $3$-approximation algorithm described in
    \secref{1_dim_3_approx}.

    \smallskip%
    \item \AprxTC: The algorithm of \AprxT with a post-\SpellIgnore{proc\-essing} stage removing unnecessary squares in the cover using greedy \SetCover to compute a smaller sub-cover.

    \smallskip%
    \item \IP: A solution to the natural integer program for computing the cover of the associated set-system of points and squares.
    The \IP was solved using \Gurobi, with a time limit of 240 seconds. \Gurobi is very good at applying various techniques in solving the \IP. In particular, the difficulty \Gurobi had in solving even mildly sized instances suggests that computing the \minWSPD is hard even for $\eps=1$ in one dimension. Nevertheless, \Gurobi seems to have a good heuristic for solving this specific \IP, and it gets a decent upper bound quite quickly.

    \smallskip%
    \item Disjoint \IP: A solution to the natural integer program for computing the \emph{disjoint} cover of the associated set-system of points and \emph{rectangles}. This \IP is much harder to solve than the original \IP, and \Gurobi had quite a hard time solving it.  Note that in this case, when the upper and lower bounds provided by \Gurobi disagree, neither bound necessarily represents a valid solution.
\end{compactenumI}

\begin{figure}[h]
    \phantom{}\hfill%
    \includegraphics[page=10,width=0.4\linewidth]%
    {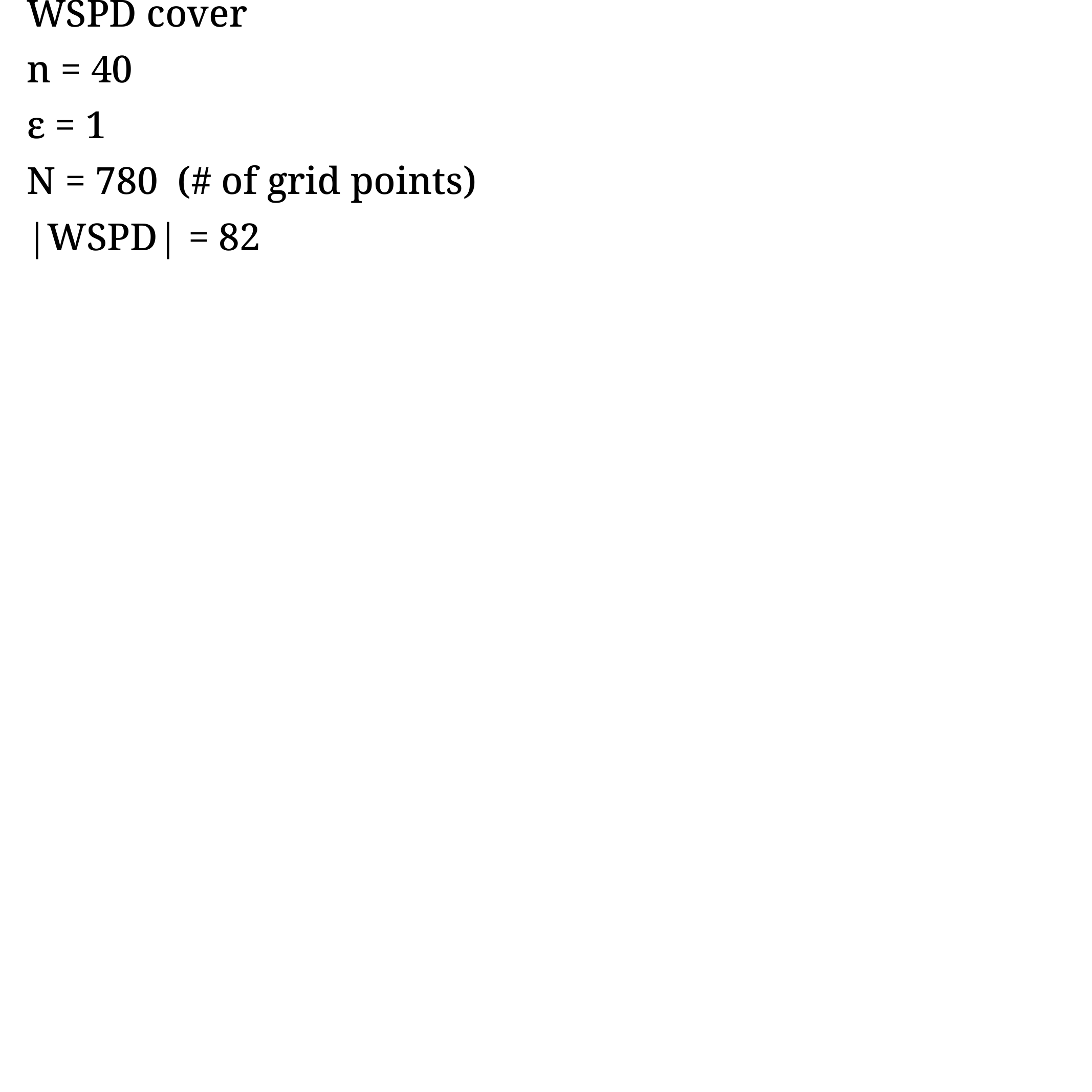}
    \hfill%
    \includegraphics[page=12,width=0.4\linewidth]%
    {figs/uniform/1_dim_n40_eps_1_0}
    \hfill\phantom{}%
    \caption{An example where the \IP and the Disjoint \IP solutions are different. Here $n=40$, the \IP solution uses $63$ squares, while the disjoint solution uses $65$. (The  \WSPD solution in this case is much worse and has size $82$.) }
    \figlab{diff}
\end{figure}

\medskip%
An example where the \IP cover and the Disjoint \IP cover have different sizes is shown in \figref{diff}.  We tried various distributions of points. Initially, it seems believable that the uniform case (i.e., the points are $1,\ldots, n$) is the hardest (see \tblref{1_dim_uniform}), but this turned out to be false. The distribution that required the largest number of pairs, that we came up with, is the logarithmic distribution (see \tblref{1_dim_log}), where the $i$\th point is $\log(i+1)$, for $i=1,\ldots, n$.  See \figref{harder} for comparison of the optimal solution in both cases.

\begin{figure}[ht]
    \phantom{}\hfill%
    \includegraphics[page=12,width=0.4\linewidth]%
    {figs/uniform/1_dim_n40_eps_1_0}
    \hfill%
    \includegraphics[page=12,width=0.4\linewidth]%
    {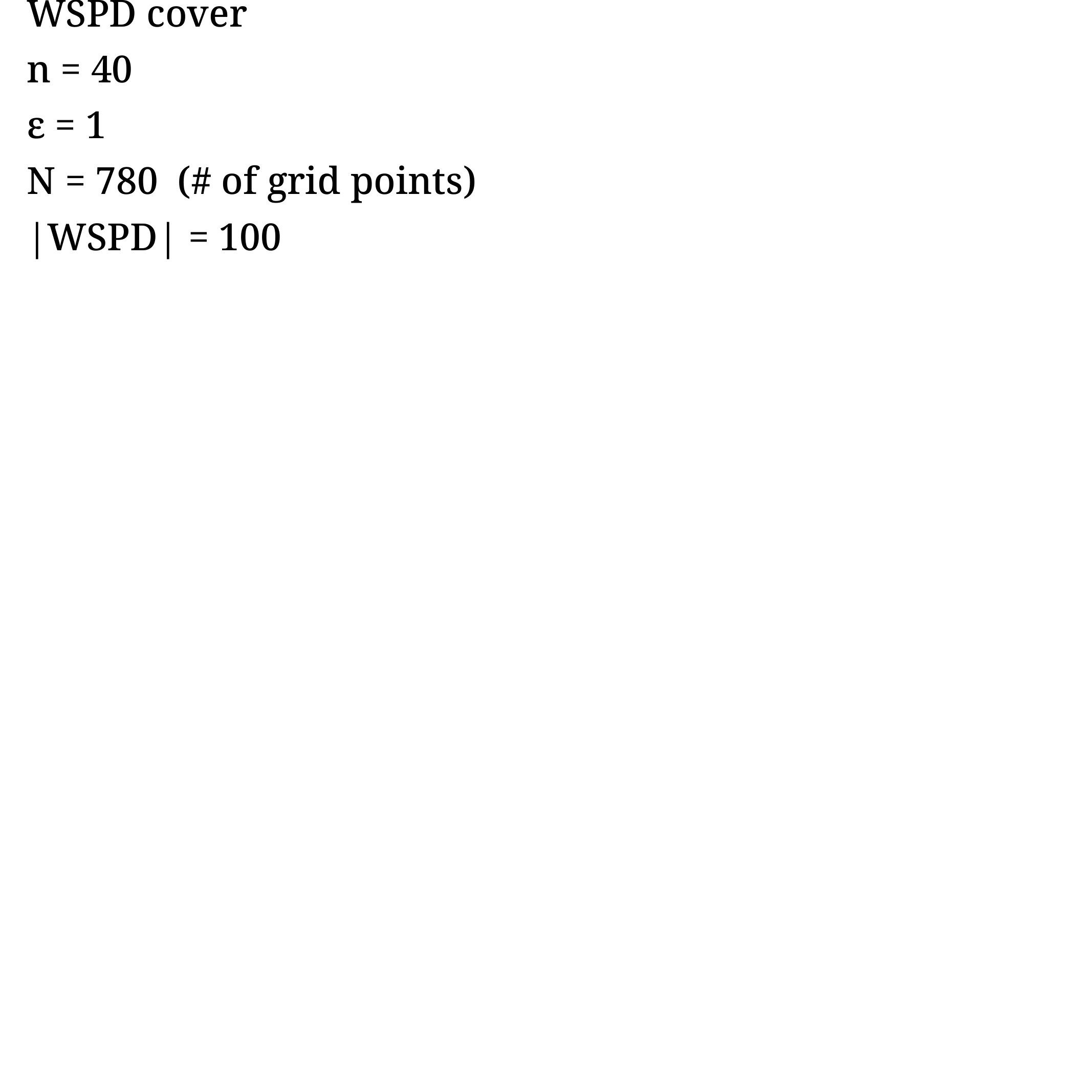}
    \phantom{}\hfill%
    \caption{The uniform $\{1,\ldots, n\}$ (left) vs the $\log$ distribution (right). The solutions shown are the Disjoint \IP covers, where the uniform case uses $63$ rectangles, while the right uses $74$. The $\log$ distribution seems to be harder because the elements close to the diameter are harder to cover and require smaller sets.  }
   \figlab{harder}
\end{figure}

See Figures \figrefX{sol_10}, \figrefX{sol_10}, \figrefX{sol_20}, \figrefX{sol_30}, and \figrefX{sol_40} for comparison of the different solutions on uniform point sets of size $n=10,20,30,40$.
We are not showing the \IP disjoint solutions because they are similar to the other solutions.
See also Tables \tblrefX{1}, \tblrefX{2}, \tblrefX{3}, \tblrefX{4}, \tblrefX{5}, and \tblrefX{6}.

\begin{remark*}
    We do not know what the worst input is for $n$ points in one dimension as far as computing \minWSPD, and it is an interesting open problem for further research.
\end{remark*}

\InsertFigure{10}{sol_10}{15}{14}{12}{12}{12}
\InsertFigure{20}{sol_20}{36}{38}{31}{30}{29}
\InsertFigure{30}{sol_30}{55}{63}{51}{49}{45}
\InsertFigure{40}{sol_40}{82}{86}{72}{69}{63}

\InsertFigure{60}{sol_60}{122}{139}{113}{108}{98}

\InsertFigure{80}{sol_80}{177}{189}{156}{146}{134}

\begin{table}[hp]
\centering
\begin{adjustbox}{max width=0.95\textwidth}
    \begin{tabular}{{ |r|r|r|r|r|r|r|r|r|r| }}
        \hline
        n & N & Greedy & WSPD & Aprx3 & Aprx3C & IP & IP$_{LB}$ & IP DJ & IP DJ$_{LB}$ \\
         \hline
        \hline
          $10$ &   $45$ &   $12$ &   $15$ &   $14$ &   $12$ &   $12$ &   $12$ &   $12$ &   $12$
         \\
        \hline
          $20$ &  $190$ &   $30$ &   $36$ &   $38$ &   $31$ &   $29$ &   $29$ &   $29$ &   $29$
         \\
        \hline
          $30$ &  $435$ &   $49$ &   $55$ &   $63$ &   $51$ &   $45$ &   $45$ &   $47$ &   $47$
         \\
        \hline
          $40$ &  $780$ &   $69$ &   $82$ &   $86$ &   $72$ &   $63$ &   $63$ &   $65$ &   $65$
         \\
        \hline
          $50$ & $1,225$ &   $87$ &  $111$ &  $112$ &   $91$ &   $81$ &   $81$ &   $84$ &   $82$
         \\
        \hline
          $60$ & $1,770$ &  $108$ &  $122$ &  $139$ &  $113$ &   $98$ &   $98$ &  $105$ &  $100$
         \\
        \hline
          $70$ & $2,415$ &  $126$ &  $143$ &  $162$ &  $132$ &  $116$ &  $115$ &  $122$ &  $118$
         \\
        \hline
          $80$ & $3,160$ &  $146$ &  $177$ &  $189$ &  $156$ &  $134$ &  $132$ & -- & --
         \\
        \hline
          $90$ & $4,005$ &  $167$ &  $207$ &  $216$ &  $177$ &  $152$ &  $149$ & -- & --
         \\
        \hline
         $100$ & $4,950$ &  $188$ &  $236$ &  $239$ &  $196$ &  $170$ &  $166$ & -- & --
         \\
        \hline
         $120$ & $7,140$ &  $228$ &  $259$ &  $293$ &  $240$ &  $208$ &  $201$ & -- & --
         \\
        \hline
         $140$ & $9,730$ &  $265$ &  $302$ &  $344$ &  $283$ &  $242$ &  $236$ & -- & --
         \\
        \hline
         $160$ & $12,720$ &  $308$ &  $368$ &  $395$ &  $325$ &  $285$ &  $271$ & -- & --
         \\
        \hline
         $180$ & $16,110$ &  $345$ &  $430$ &  $450$ &  $371$ &  $335$ & -- & -- & --
         \\
        \hline
         $200$ & $19,900$ &  $387$ &  $489$ &  $501$ &  $414$ &  $374$ & -- & -- & --
         \\
        \hline
         $250$ & $31,125$ &  $501$ &  $546$ &  $630$ &  $520$ &  $473$ & -- & -- & --
         \\
        \hline
         $300$ & $44,850$ &  $591$ &  $692$ &  $763$ &  $630$ &  $573$ & -- & -- & --
         \\
        \hline
         $350$ & $61,075$ &  $696$ &  $837$ &  $893$ &  $740$ &  $672$ & -- & -- & --
         \\
        \hline
         $400$ & $79,800$ &  $797$ &  $998$ & $1,023$ &  $848$ &  $771$ & -- & -- & --
         \\
        \hline
    \end{tabular}
\end{adjustbox}

    \caption{Results for the $1$-dim for the point set $1,2, \ldots, n$ and $\eps=1$. Here, $n$ is the number of points, and $N = n (n-1)/2$ is the number in the 2d lifted point set. The $LB$ designation refers to the lower bound.}
    \tbllab{1_dim_uniform}
    \tbllab{1}
\end{table}

\begin{table}[p]
\centering
\begin{adjustbox}{max width=0.95\textwidth}
    \begin{tabular}{{ |r|r|r|r|r|r|r|r|r|r| }}
        \hline
        n & N & Greedy & WSPD & Aprx3 & Aprx3C & IP & IP$_{LB}$ & IP DJ & IP DJ$_{LB}$ \\
         \hline
        \hline
          $10$ &   $45$ &   $10$ &   $10$ &   $14$ &   $11$ &   $10$ &   $10$ &   $10$ &   $10$
         \\
        \hline
          $20$ &  $190$ &   $27$ &   $30$ &   $37$ &   $31$ &   $27$ &   $27$ &   $27$ &   $27$
         \\
        \hline
          $30$ &  $435$ &   $50$ &   $62$ &   $65$ &   $52$ &   $45$ &   $45$ &   $45$ &   $45$
         \\
        \hline
          $40$ &  $780$ &   $63$ &   $74$ &   $85$ &   $70$ &   $59$ &   $59$ &   $60$ &   $60$
         \\
        \hline
          $50$ & $1,225$ &   $76$ &   $89$ &   $99$ &   $88$ &   $72$ &   $72$ &   $72$ &   $72$
         \\
        \hline
          $60$ & $1,770$ &   $97$ &  $120$ &  $128$ &  $109$ &   $91$ &   $91$ &   $92$ &   $92$
         \\
        \hline
          $70$ & $2,415$ &  $119$ &  $140$ &  $157$ &  $129$ &  $108$ &  $108$ &  $109$ &  $109$
         \\
        \hline
          $80$ & $3,160$ &  $146$ &  $163$ &  $184$ &  $159$ &  $130$ &  $130$ & -- & --
         \\
        \hline
          $90$ & $4,005$ &  $163$ &  $184$ &  $210$ &  $170$ &  $142$ &  $142$ & -- & --
         \\
        \hline
         $100$ & $4,950$ &  $171$ &  $203$ &  $227$ &  $191$ &  $158$ &  $158$ & -- & --
         \\
        \hline
         $120$ & $7,140$ &  $203$ &  $235$ &  $266$ &  $221$ &  $183$ &  $183$ & -- & --
         \\
        \hline
         $140$ & $9,730$ &  $247$ &  $285$ &  $319$ &  $273$ &  $223$ &  $223$ & -- & --
         \\
        \hline
         $160$ & $12,720$ &  $277$ &  $318$ &  $373$ &  $317$ &  $256$ &  $256$ & -- & --
         \\
        \hline
         $180$ & $16,110$ &  $315$ &  $378$ &  $398$ &  $345$ &  $282$ &  $282$ & -- & --
         \\
        \hline
         $200$ & $19,900$ &  $359$ &  $416$ &  $470$ &  $398$ &  $324$ &  $322$ & -- & --
         \\
        \hline
         $250$ & $31,125$ &  $438$ &  $515$ &  $565$ &  $476$ &  $424$ & -- & -- & --
         \\
        \hline
         $300$ & $44,850$ &  $548$ &  $631$ &  $693$ &  $590$ &  $527$ & -- & -- & --
         \\
        \hline
         $350$ & $61,075$ &  $655$ &  $748$ &  $840$ &  $724$ &  $625$ & -- & -- & --
         \\
        \hline
         $400$ & $79,800$ &  $737$ &  $823$ &  $949$ &  $814$ &  $696$ & -- & -- & --
         \\
        \hline
    \end{tabular}
\end{adjustbox}

    \caption{Results for the $1$-dim for a set of $n$ points picked uniformly from an interval, with $\eps=1$.}
    \tbllab{1_dim_u_random}
    \tbllab{2}
\end{table}

\begin{table}
\centering
\begin{adjustbox}{max width=0.95\textwidth}
    \begin{tabular}{{ |r|r|r|r|r|r|r|r|r|r| }}
        \hline
        n & N & Greedy & WSPD & Aprx3 & Aprx3C & IP & IP$_{LB}$ & IP DJ & IP DJ$_{LB}$ \\
         \hline
        \hline
          $10$ &   $45$ &   $12$ &   $12$ &   $13$ &   $12$ &   $12$ &   $12$ &   $12$ &   $12$
         \\
        \hline
          $20$ &  $190$ &   $29$ &   $30$ &   $34$ &   $32$ &   $29$ &   $29$ &   $29$ &   $29$
         \\
        \hline
          $30$ &  $435$ &   $50$ &   $61$ &   $67$ &   $57$ &   $46$ &   $46$ &   $47$ &   $47$
         \\
        \hline
          $40$ &  $780$ &   $65$ &   $80$ &   $83$ &   $71$ &   $61$ &   $61$ &   $62$ &   $62$
         \\
        \hline
          $50$ & $1,225$ &   $82$ &   $91$ &  $100$ &   $89$ &   $75$ &   $75$ &   $75$ &   $75$
         \\
        \hline
          $60$ & $1,770$ &  $106$ &  $114$ &  $132$ &  $114$ &   $93$ &   $93$ &   $94$ &   $94$
         \\
        \hline
          $70$ & $2,415$ &  $122$ &  $133$ &  $156$ &  $133$ &  $109$ &  $109$ &  $110$ &  $110$
         \\
        \hline
          $80$ & $3,160$ &  $142$ &  $154$ &  $186$ &  $152$ &  $127$ &  $127$ & -- & --
         \\
        \hline
          $90$ & $4,005$ &  $152$ &  $176$ &  $207$ &  $179$ &  $143$ &  $143$ & -- & --
         \\
        \hline
         $100$ & $4,950$ &  $174$ &  $205$ &  $219$ &  $190$ &  $158$ &  $158$ & -- & --
         \\
        \hline
         $120$ & $7,140$ &  $216$ &  $249$ &  $292$ &  $243$ &  $192$ &  $192$ & -- & --
         \\
        \hline
         $140$ & $9,730$ &  $257$ &  $286$ &  $324$ &  $284$ &  $228$ &  $227$ & -- & --
         \\
        \hline
         $160$ & $12,720$ &  $284$ &  $318$ &  $373$ &  $317$ &  $259$ &  $258$ & -- & --
         \\
        \hline
         $180$ & $16,110$ &  $319$ &  $382$ &  $419$ &  $357$ &  $290$ &  $290$ & -- & --
         \\
        \hline
         $200$ & $19,900$ &  $366$ &  $427$ &  $461$ &  $395$ &  $323$ &  $322$ & -- & --
         \\
        \hline
         $250$ & $31,125$ &  $445$ &  $518$ &  $581$ &  $499$ &  $428$ & -- & -- & --
         \\
        \hline
         $300$ & $44,850$ &  $554$ &  $623$ &  $743$ &  $626$ &  $529$ & -- & -- & --
         \\
        \hline
         $350$ & $61,075$ &  $651$ &  $749$ &  $832$ &  $713$ &  $618$ & -- & -- & --
         \\
        \hline
         $400$ & $79,800$ &  $757$ &  $883$ &  $980$ &  $847$ &  $712$ & -- & -- & --
         \\
        \hline
    \end{tabular}
\end{adjustbox}

    \caption{Results for the $1$-dim for a set of $n$ points picked according to the normal distribution, with $\eps=1$.}
    \tbllab{1_dim_normal}
    \tbllab{3}
\end{table}

\begin{table}
\centering
\begin{adjustbox}{max width=0.95\textwidth}
    \begin{tabular}{{ |r|r|r|r|r|r|r|r|r|r| }}
        \hline
        n & N & Greedy & WSPD & Aprx3 & Aprx3C & IP & IP$_{LB}$ & IP DJ & IP DJ$_{LB}$ \\
         \hline
        \hline
          $10$ &   $45$ &   $13$ &   $16$ &   $20$ &   $16$ &   $13$ &   $13$ &   $13$ &   $13$
         \\
        \hline
          $20$ &  $190$ &   $34$ &   $41$ &   $49$ &   $38$ &   $32$ &   $32$ &   $32$ &   $32$
         \\
        \hline
          $30$ &  $435$ &   $54$ &   $69$ &   $80$ &   $62$ &   $52$ &   $52$ &   $52$ &   $52$
         \\
        \hline
          $40$ &  $780$ &   $78$ &   $96$ &  $110$ &   $91$ &   $71$ &   $71$ &   $71$ &   $71$
         \\
        \hline
          $50$ & $1,225$ &   $97$ &  $122$ &  $142$ &  $115$ &   $91$ &   $91$ &   $92$ &   $92$
         \\
        \hline
          $60$ & $1,770$ &  $120$ &  $150$ &  $175$ &  $143$ &  $111$ &  $111$ &  $112$ &  $112$
         \\
        \hline
          $70$ & $2,415$ &  $142$ &  $176$ &  $208$ &  $169$ &  $131$ &  $131$ &  $133$ &  $132$
         \\
        \hline
          $80$ & $3,160$ &  $164$ &  $208$ &  $241$ &  $190$ &  $152$ &  $152$ & -- & --
         \\
        \hline
          $90$ & $4,005$ &  $186$ &  $235$ &  $274$ &  $218$ &  $172$ &  $172$ & -- & --
         \\
        \hline
         $100$ & $4,950$ &  $211$ &  $258$ &  $307$ &  $244$ &  $193$ &  $192$ & -- & --
         \\
        \hline
         $120$ & $7,140$ &  $254$ &  $320$ &  $374$ &  $299$ &  $236$ &  $232$ & -- & --
         \\
        \hline
         $140$ & $9,730$ &  $298$ &  $376$ &  $441$ &  $350$ &  $280$ &  $272$ & -- & --
         \\
        \hline
         $160$ & $12,720$ &  $347$ &  $436$ &  $509$ &  $404$ &  $320$ &  $312$ & -- & --
         \\
        \hline
         $180$ & $16,110$ &  $392$ &  $490$ &  $577$ &  $465$ &  $365$ &  $353$ & -- & --
         \\
        \hline
         $200$ & $19,900$ &  $432$ &  $542$ &  $645$ &  $518$ &  $417$ & -- & -- & --
         \\
        \hline
         $250$ & $31,125$ &  $548$ &  $688$ &  $815$ &  $646$ &  $523$ & -- & -- & --
         \\
        \hline
         $300$ & $44,850$ &  $664$ &  $832$ &  $986$ &  $787$ &  $629$ & -- & -- & --
         \\
        \hline
         $350$ & $61,075$ &  $783$ &  $975$ & $1,157$ &  $921$ &  $737$ & -- & -- & --
         \\
        \hline
         $400$ & $79,800$ &  $893$ & $1,112$ & $1,328$ & $1,060$ &  $848$ & -- & -- & --
         \\
        \hline
    \end{tabular}
\end{adjustbox}

    \caption{Results for the $1$-dim for a set of $n$ points, where the $i$\th point is $i^2$, where $i=1,\ldots, n$, with $\eps=1$.}
    \tbllab{1_dim_squares}
    \tbllab{4}
\end{table}

\begin{table}
\centering
\begin{adjustbox}{max width=0.95\textwidth}
    \begin{tabular}{{ |r|r|r|r|r|r|r|r|r|r| }}
        \hline
        n & N & Greedy & WSPD & Aprx3 & Aprx3C & IP & IP$_{LB}$ & IP DJ & IP DJ$_{LB}$ \\
         \hline
        \hline
          $10$ &   $45$ &   $12$ &   $15$ &   $19$ &   $15$ &   $12$ &   $12$ &   $12$ &   $12$
         \\
        \hline
          $20$ &  $190$ &   $31$ &   $38$ &   $46$ &   $38$ &   $31$ &   $31$ &   $31$ &   $31$
         \\
        \hline
          $30$ &  $435$ &   $52$ &   $64$ &   $75$ &   $62$ &   $50$ &   $50$ &   $50$ &   $50$
         \\
        \hline
          $40$ &  $780$ &   $72$ &   $90$ &  $106$ &   $86$ &   $70$ &   $70$ &   $70$ &   $70$
         \\
        \hline
          $50$ & $1,225$ &   $94$ &  $120$ &  $137$ &  $112$ &   $89$ &   $89$ &   $89$ &   $89$
         \\
        \hline
          $60$ & $1,770$ &  $114$ &  $145$ &  $169$ &  $138$ &  $109$ &  $109$ &  $109$ &  $109$
         \\
        \hline
          $70$ & $2,415$ &  $135$ &  $171$ &  $199$ &  $168$ &  $129$ &  $129$ &  $131$ &  $130$
         \\
        \hline
          $80$ & $3,160$ &  $163$ &  $203$ &  $231$ &  $190$ &  $149$ &  $149$ & -- & --
         \\
        \hline
          $90$ & $4,005$ &  $180$ &  $226$ &  $263$ &  $215$ &  $170$ &  $169$ & -- & --
         \\
        \hline
         $100$ & $4,950$ &  $204$ &  $259$ &  $297$ &  $242$ &  $190$ &  $188$ & -- & --
         \\
        \hline
         $120$ & $7,140$ &  $247$ &  $316$ &  $362$ &  $295$ &  $230$ &  $228$ & -- & --
         \\
        \hline
         $140$ & $9,730$ &  $293$ &  $368$ &  $429$ &  $345$ &  $274$ &  $268$ & -- & --
         \\
        \hline
         $160$ & $12,720$ &  $343$ &  $427$ &  $496$ &  $399$ &  $315$ &  $309$ & -- & --
         \\
        \hline
         $180$ & $16,110$ &  $383$ &  $482$ &  $562$ &  $450$ &  $362$ &  $349$ & -- & --
         \\
        \hline
         $200$ & $19,900$ &  $427$ &  $541$ &  $629$ &  $505$ &  $406$ &  $389$ & -- & --
         \\
        \hline
         $250$ & $31,125$ &  $543$ &  $686$ &  $799$ &  $640$ &  $519$ & -- & -- & --
         \\
        \hline
         $300$ & $44,850$ &  $654$ &  $822$ &  $968$ &  $771$ &  $628$ & -- & -- & --
         \\
        \hline
         $350$ & $61,075$ &  $774$ &  $962$ & $1,138$ &  $904$ &  $739$ & -- & -- & --
         \\
        \hline
         $400$ & $79,800$ &  $888$ & $1,112$ & $1,309$ & $1,045$ &  $848$ & -- & -- & --
         \\
        \hline
    \end{tabular}
\end{adjustbox}

    \caption{Results for the $1$-dim for a set of $n$ points, where the $i$\th point is $i^3$, where $i=1,\ldots, n$, with $\eps=1$.}
    \tbllab{1_dim_cubes}
    \tbllab{5}
\end{table}

\begin{table}
\centering
\begin{adjustbox}{max width=0.95\textwidth}
    \begin{tabular}{{ |r|r|r|r|r|r|r|r|r|r| }}
        \hline
        n & N & Greedy & WSPD & Aprx3 & Aprx3C & IP & IP$_{LB}$ & IP DJ & IP DJ$_{LB}$ \\
         \hline
        \hline
          $10$ &   $45$ &   $14$ &   $18$ &   $26$ &   $15$ &   $14$ &   $14$ &   $14$ &   $14$
         \\
        \hline
          $20$ &  $190$ &   $37$ &   $44$ &   $68$ &   $42$ &   $33$ &   $33$ &   $34$ &   $34$
         \\
        \hline
          $30$ &  $435$ &   $59$ &   $71$ &  $108$ &   $67$ &   $53$ &   $53$ &   $54$ &   $54$
         \\
        \hline
          $40$ &  $780$ &   $82$ &  $100$ &  $154$ &   $95$ &   $73$ &   $73$ &   $74$ &   $74$
         \\
        \hline
          $50$ & $1,225$ &   $99$ &  $127$ &  $194$ &  $122$ &   $93$ &   $93$ &   $94$ &   $94$
         \\
        \hline
          $60$ & $1,770$ &  $126$ &  $156$ &  $238$ &  $150$ &  $114$ &  $114$ &  $115$ &  $115$
         \\
        \hline
          $70$ & $2,415$ &  $138$ &  $185$ &  $280$ &  $178$ &  $134$ &  $134$ &  $138$ &  $135$
         \\
        \hline
          $80$ & $3,160$ &  $174$ &  $211$ &  $326$ &  $207$ &  $155$ &  $155$ & -- & --
         \\
        \hline
          $90$ & $4,005$ &  $195$ &  $241$ &  $369$ &  $237$ &  $176$ &  $175$ & -- & --
         \\
        \hline
         $100$ & $4,950$ &  $213$ &  $269$ &  $411$ &  $261$ &  $197$ &  $195$ & -- & --
         \\
        \hline
         $120$ & $7,140$ &  $259$ &  $326$ &  $499$ &  $320$ &  $240$ &  $235$ & -- & --
         \\
        \hline
         $140$ & $9,730$ &  $310$ &  $384$ &  $587$ &  $376$ &  $285$ &  $276$ & -- & --
         \\
        \hline
         $160$ & $12,720$ &  $355$ &  $437$ &  $672$ &  $431$ &  $326$ &  $317$ & -- & --
         \\
        \hline
         $180$ & $16,110$ &  $411$ &  $497$ &  $760$ &  $490$ &  $375$ &  $358$ & -- & --
         \\
        \hline
         $200$ & $19,900$ &  $452$ &  $558$ &  $850$ &  $549$ &  $449$ & -- & -- & --
         \\
        \hline
         $250$ & $31,125$ &  $564$ &  $693$ & $1,066$ &  $692$ &  $569$ & -- & -- & --
         \\
        \hline
         $300$ & $44,850$ &  $686$ &  $841$ & $1,287$ &  $836$ &  $686$ & -- & -- & --
         \\
        \hline
         $350$ & $61,075$ &  $819$ &  $989$ & $1,506$ &  $981$ &  $803$ & -- & -- & --
         \\
        \hline
         $400$ & $79,800$ &  $925$ & $1,125$ & $1,722$ & $1,123$ &  $923$ & -- & -- & --
         \\
        \hline
    \end{tabular}
\end{adjustbox}

    \caption{Results for the $1$-dim for a set of $n$ points, where the $i$\th point is $\log(i+1)$, where $i=1,\ldots, n$, with $\eps=1$.}
    \tbllab{1_dim_log}
    \tbllab{6}
\end{table}

\begin{table}
\centering
\begin{adjustbox}{max width=0.95\textwidth}
    \begin{tabular}{{ |r|r|r|r|r|r|r|r|r|r| }}
        \hline
        n & N & Greedy & WSPD & Aprx3 & Aprx3C & IP & IP$_{LB}$ & IP DJ & IP DJ$_{LB}$ \\
         \hline
        \hline
          $10$ &   $45$ &   $15$ &   $17$ &   $28$ &   $16$ &   $14$ &   $14$ &   $14$ &   $14$
         \\
        \hline
          $20$ &  $190$ &   $38$ &   $43$ &   $68$ &   $40$ &   $33$ &   $33$ &   $33$ &   $33$
         \\
        \hline
          $30$ &  $435$ &   $60$ &   $71$ &  $112$ &   $68$ &   $53$ &   $53$ &   $54$ &   $54$
         \\
        \hline
          $40$ &  $780$ &   $85$ &   $99$ &  $153$ &   $96$ &   $73$ &   $73$ &   $74$ &   $74$
         \\
        \hline
          $50$ & $1,225$ &  $102$ &  $127$ &  $195$ &  $122$ &   $93$ &   $93$ &   $94$ &   $94$
         \\
        \hline
          $60$ & $1,770$ &  $125$ &  $155$ &  $239$ &  $151$ &  $114$ &  $114$ &  $118$ &  $114$
         \\
        \hline
          $70$ & $2,415$ &  $148$ &  $181$ &  $282$ &  $178$ &  $135$ &  $135$ &  $137$ &  $135$
         \\
        \hline
          $80$ & $3,160$ &  $173$ &  $207$ &  $326$ &  $206$ &  $155$ &  $155$ & -- & --
         \\
        \hline
          $90$ & $4,005$ &  $203$ &  $240$ &  $368$ &  $232$ &  $176$ &  $175$ & -- & --
         \\
        \hline
         $100$ & $4,950$ &  $221$ &  $270$ &  $411$ &  $263$ &  $197$ &  $195$ & -- & --
         \\
        \hline
         $120$ & $7,140$ &  $267$ &  $323$ &  $499$ &  $320$ &  $240$ &  $235$ & -- & --
         \\
        \hline
         $140$ & $9,730$ &  $310$ &  $381$ &  $587$ &  $376$ &  $283$ &  $276$ & -- & --
         \\
        \hline
         $160$ & $12,720$ &  $358$ &  $441$ &  $674$ &  $437$ &  $327$ &  $317$ & -- & --
         \\
        \hline
         $180$ & $16,110$ &  $409$ &  $500$ &  $761$ &  $491$ &  $391$ & -- & -- & --
         \\
        \hline
         $200$ & $19,900$ &  $449$ &  $553$ &  $848$ &  $548$ &  $437$ & -- & -- & --
         \\
        \hline
         $250$ & $31,125$ &  $569$ &  $691$ & $1,069$ &  $693$ &  $549$ & -- & -- & --
         \\
        \hline
         $300$ & $44,850$ &  $696$ &  $841$ & $1,287$ &  $837$ &  $663$ & -- & -- & --
         \\
        \hline
         $350$ & $61,075$ &  $801$ &  $982$ & $1,505$ &  $980$ &  $776$ & -- & -- & --
         \\
        \hline
         $400$ & $79,800$ &  $925$ & $1,119$ & $1,725$ & $1,128$ &  $893$ & -- & -- & --
         \\
        \hline
    \end{tabular}
\end{adjustbox}

    \caption{Results for the $1$-dim for a set of $n$ points, where the $i$\th point is $\bigl(\log(i+1)\bigr)^2$, where $i=1,\ldots, n$, with $\eps=1$.}
    \tbllab{1_dim_log_sq}
\end{table}

\FloatBarrier%
\section{Conclusions}%
\seclab{conclusions}

We introduced a new pair decomposition, an \ABC, and showed that every finite metric space admits a near-linear-sized \ABC. We also showed that Euclidean spaces admit a linear-sized \ABC, even in high dimensions (the constant there is $O(d^3)$.  In low-dimensional Euclidean space, these decompositions have small implicit representation, similar to \WSPD.  Some open problems include: are our bounds tight in \thmref{abc_metric} and \thmref{abc_r_d}? Which metric spaces, other than $\Re^d$, admit a linear-sized \ABC? We consider the existence of the \ABC as a surprising property, especially in general metric spaces. We believe that the \ABC could lead to further structural or algorithmic insights in the future.

We provided a number of results for approximating the minimum \WSPD.  Some open problems include: Is the minimum \WSPD problem \NPHard in $\Re$? Is there a $(1+\eps)$-approximation algorithm, or a simple $O(1)$-approximation algorithm, in $\Re^d$, for $d \geq 2$? What are the worst-case inputs for \minWSPD in $\Re^d$?

\begingroup%
\emergencystretch=1em
\printbibliography{}%
\endgroup %

\appendix

\section{Tedious calculations}

\begin{tedium}
    \tedlab{boring_1}%
    We  have that
    \begin{align*}
        \fstabY{\PB_t}{\PC_t}
      &=
        \frac{\dmaxY{\PB_t}{\PC_t}- \dmY{\PB_t}{\PC_t}}{2\dmY{\PB_t}{\PC_t}}
      \\&%
        \leq
        \frac{(1+2/n^{3})
           \Delta_0 -
           (1-4/n^{3}) \dmY{\PB_0}{\PC_0}}
        {2(1-4/n^{3}) \dmY{\PB_0}{\PC_0}}
        \\&
        =
        \frac{\frac{1+2/n^{3}}{1-4/n^3}
           \Delta_0 -
            \dmY{\PB_0}{\PC_0}}
        {2 \dmY{\PB_0}{\PC_0}}
      =
      \frac{\frac{6/n^{3}}{1-4/n^3}
           \Delta_0  + \Delta_0-
            \dmY{\PB_0}{\PC_0}}
        {2 \dmY{\PB_0}{\PC_0}}
      \\&%
      \leq
      \fstabY{\PB_0}{\PC_0} +
      \frac{6/n^{3}}{1-4/n^3} \cdot
      \frac{\Delta_0}{2 \dmY{\PB_0}{\PC_0}}
      \leq
      \frac{\eps}{16} + \frac{12}{n^3}
      \leq
      \frac{\eps}{8},
    \end{align*}
    since $\Delta_0 / \dmY{\PB_0}{\PC_0} \leq 1 + 2\fstabY{\PB_0}{\PC_0}$ and $n \geq 1/\eps$.
\end{tedium}

    \begin{tedium}
        \tedlab{boring_2}%
        Indeed, $\diamX{\Cell} = \sqrt{d} \num$, and let $Q = P \cap \Cell$. By construction, we have that
        \begin{align*}
          u &= \dminY{Q}{P_{v,i}}
              \geq
              \rho + i \num - 2\sqrt{d}\num
          \\
          \text{and}\qquad%
          U &= \dmaxY{Q}{P_{v,i}}
              \leq
              \rho + i \num + 2\sqrt{d}\num.
        \end{align*}
        As such, the stability of this pair is
        \begin{align*}
          \fstabY{Q}{P_{v,i}}
          &=
            \frac{U - u}{2u}
            \leq%
            \frac{4\sqrt{d}\num}{2(\rho + i \num - 2\sqrt{d}\num)}
            =
            \frac{4\sqrt{d}\num}{2(\frac{8 \sqrt{d}}{\eps} \num + i \num - 2\sqrt{d}\num)}
          \\&
            \leq
            \frac{4\sqrt{d}}{2(\frac{8 \sqrt{d}}{\eps}  - 2\sqrt{d})}
          \\&%
          =
          \frac{4}{2(\frac{8 }{\eps}  - 2)}
          \leq
          \frac{4}{2\cdot \frac{6 }{\eps}}
          =
          \frac{\eps}{3}.
        \end{align*}
    \end{tedium}

\end{document}